\documentclass{ieeeaccess}
\usepackage{cite}
\usepackage{amsmath,amssymb,amsfonts}
\usepackage{algorithmic}
\usepackage{graphicx}
\usepackage{textcomp}
\newcommand{\st}{s\textsf{-}t}
\usepackage{amsmath,mathtools,bm,amssymb}
\usepackage{enumitem}
\usepackage{xspace}
\usepackage{pifont}
\usepackage{dsfont}
\usepackage{graphicx}
\usepackage{subfig}
\usepackage{grffile}
\usepackage{pdflscape}
\usepackage{breqn}
\usepackage{amsthm}
\usepackage{bm}
\usepackage{color}
\usepackage{physics}
\allowdisplaybreaks
\graphicspath{{./figs/}}
\pdfoutput=1

\newcommand{\gv}[1]{\textcolor{magenta}{#1}}

\newtheorem{lemma}{Lemma}

\newcommand{\pr}[1]{#1^{\prime}}
\newcommand{\eg}{\textit{e.g.}}
\newcommand{\ie}{\textit{i.e.}}
\newcommand{\etal}{\textit{et al.}\xspace}
\newcommand{\etc}{\textit{etc.}}

\def\BibTeX{{\rm B\kern-.05em{\sc i\kern-.025em b}\kern-.08em
    T\kern-.1667em\lower.7ex\hbox{E}\kern-.125emX}}
\begin{document}
\doi{}

\title{On the Bipartite Entanglement Capacity of Quantum Networks}
\author{\uppercase{Gayane Vardoyan}\authorrefmark{1,3}, 
\uppercase{Emily van Milligen}\authorrefmark{2},
\uppercase{Saikat Guha}\authorrefmark{2},
\uppercase{Stephanie Wehner},\authorrefmark{1}
and \uppercase{Don Towsley}\authorrefmark{3}}
\address[1]{Delft University of Technology (emails: g.s.vardoyan@tudelft.nl, s.d.c.wehner@tudelft.nl)}
\address[2]{The University of Arizona (emails: evanmilligen@arizona.edu, saikat@arizona.edu)}
\address[3]{University of Massachusetts, Amherst (emails: gvardoyan@cs.umass.edu, towsley@cs.umass.edu)}



\begin{abstract}
We consider the problem of multi-path entanglement distribution to a pair of nodes in a quantum network consisting of devices with non-deterministic entanglement swapping capabilities. 
Multi-path entanglement distribution enables a network to establish end-to-end entangled links across any number of available paths with pre-established elementary (link-level) entanglement.
Probabilistic entanglement swapping, on the other hand, limits the amount of entanglement that is shared between the nodes; this is especially the case when, due to architectural and other practical constraints, swaps must be performed in temporal proximity to each other. Limiting our focus to the case where only bipartite entangled states are generated across the network, we cast the problem as an instance of generalized flow maximization  between two quantum end nodes wishing to communicate.
We propose a mixed-integer quadratically constrained program (MIQCP) to solve this flow problem for networks with arbitrary topology. We then compute the overall network capacity, defined as the maximum number of EPR states distributed to users per time unit, by solving the flow problem for all possible network states generated by probabilistic entangled link presence and absence, and subsequently by averaging over all network state capacities. The MIQCP can also be applied to networks with multiplexed links -- \ie, links capable of generating multiple EPR pairs per time unit -- by converting the original network state's graph into an equivalent one with unit link capacities. While our approach for computing the overall network capacity has the undesirable property that the total number of states grows exponentially with link multiplexing capability, it nevertheless yields an exact solution that serves as an upper bound comparison basis for the throughput performance of more easily-implementable yet non-optimal entanglement routing algorithms. We apply our capacity computation method to several multiplexed and non-multiplexed networks, including a topology based on SURFnet -- a backbone network used for research purposes in the Netherlands -- for which real fiber length data is available.
\end{abstract}

\begin{keywords}
entanglement distribution and routing, mixed-integer quadratically constrained program, quantum network
\end{keywords}

\titlepgskip=-15pt

\maketitle

\section{Introduction}
\label{sec:introduction}
Quantum networks consisting of first- and second-generation quantum repeaters \cite{munro2015inside,muralidharan2016optimal} rely on entanglement distribution to enable distributed quantum applications between distant nodes. 
In the most basic scenario, bipartite entangled states are first established at the physical (elementary) link level between neighboring nodes, and later ``connected'' via entangling measurements. The latter, a process also known as entanglement swapping, serves to extend entanglement across a longer distance, but for certain quantum architectures succeeds only probabilistically. Figure~\ref{fig:swap_success} illustrates a successful execution of this process using the state $\ket{\Psi^+}=(\ket{00}+\ket{11})/\sqrt{2}$ as an example. In large, complex quantum networks with arbitrary topology, entanglement routing algorithms are necessary to determine the precise subset of all available swapping operations that must be performed in order to link two nodes that wish to communicate.
Efficient use of entanglement -- a resource that is created probabilistically and whose usefulness carries an expiration date -- is thus an advantageous property for such algorithms. The number of choices that confront an entanglement routing algorithm further increases when there is an option of utilizing multiple paths to distribute as many bipartite entangled states to a pair of users as possible, and when entanglement swapping success probability varies from node to node.
\begin{figure}
    \centering
    \includegraphics[width=0.98\linewidth,trim={3cm 4cm 3cm 3cm},clip]{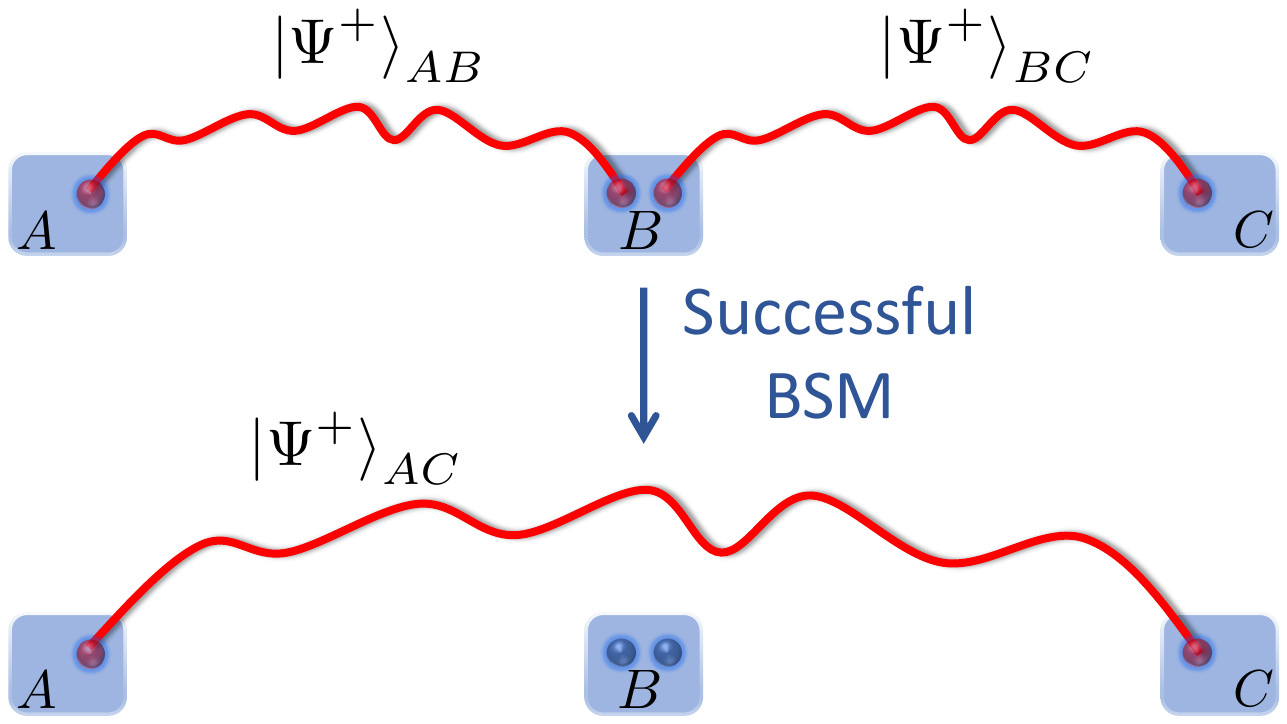}
    \caption{Entanglement swapping. Quantum network nodes $A$, $B$, and $C$ begin by generating elementary link-level entangled states $\ket{\Psi^+}_{AB}$ and $\ket{\Psi^+}_{BC}$. A successful swap/BSM (Bell state measurement) at node $B$ results in an entangled state $\ket{\Psi^+}_{AC}$ shared between $A$ and $C$.}
    \label{fig:swap_success}
\end{figure}

A metric of interest for the multi-path entanglement distribution problem is the overall \textit{bipartite entanglement capacity (BEC)} of a given network, defined here as the average number of bipartite entangled quantum states that the network can provide to a specific pair of users or nodes, $s$ and $t$ (a mathematical definition is provided later in this section). This quantity is especially relevant to applications which have either stringent rate requirements or a necessity to generate multiple contemporaneous EPR (Einstein–Podolsky–Rosen) pairs. An example of the former is a scenario wherein a pair of users carrying out a quantum key distribution (QKD) \cite{bennet1984quantum,ekert1991quantum,bennett1992quantum} protocol wish to generate a secret key by a given deadline. An example of the latter is a Blind Quantum Computation (BQC) protocol \cite{broadbent2009universal,leichtle2021verifying} wherein typically, a user and server must establish several EPR pairs close together in time, in order to be able to carry out each round of the protocol at all. The ability to utilize multiple paths along which to generate entanglement may thus not only boost an application's performance but also make it feasible for implementation on a network. Computation of the BEC provides insight into a network's ability to support such quantum applications. We note that in this work, the BEC is always defined with respect to a node pair $(s,t)$; and further, that ``capacity'' here incorporates only the rate at which end-to-end entanglement is generated  (in contrast to other definitions of capacity commonly found in quantum literature).

Our goal in this paper is to compute the BEC for an arbitrary network serving a single user pair.
We approach this problem by computing
the capacities of all possible network states that result from probabilistic elementary link\footnote{An elementary link is one that exists between two neighboring network nodes, \ie, nodes which have a direct physical connection (\eg, optical fiber).} entanglement generation, and subsequently take the expectation of these values. In this paper, ``network state'' is defined in terms of the presence and location of entangled links; this is not to be confused with \textit{quantum} states.
We formulate the (sub)problem of computing the capacity of each network state as
an $\st$ flow optimization problem, where the flow through the network may not split or merge, and the network nodes ``leak'' a certain amount of flow corresponding to their entanglement swapping failure probabilities. We cast this as a mixed-integer quadratically constrained program (MIQCP), which can be directly applied to non-multiplexed networks, \ie, those in which each link is able to generate at most one entangled state per allotted time interval. With a straightforward modification to the graph corresponding to a given network state, the MIQCP can also be applied to a multiplexed network (one in which each elementary link may generate multiple entangled states per time interval). A favourable consequence of the MIQCP approach is that its solution for a given network state yields not only the capacity of the state, but also a way to obtain it: each variable's value in the solution specifies the entanglement routing procedure, \ie, the optimal set of swaps to be performed.

A quantum network may be represented as an undirected graph $\mathcal{G}=(\mathcal{V},\mathcal{E})$ with $\mathcal{V}$ comprising a set of nodes (\eg, quantum-equipped users, quantum repeaters, switches, \etc) interconnected via physical links that make up the edge set $\mathcal{E}$. For a given point in time, the \textit{network state} is also represented by a graph, $\mathcal{G}^{\prime}=(\mathcal{V},\mathcal{E}^{\prime})$, where $\mathcal{E}^{\prime}\subseteq \mathcal{E}$ indicates successfully generated link-level entanglement. A model commonly used in literature is a network that operates in a synchronized, time-slotted manner -- \ie, time is discretized into equal-length slots wherein network state may change according to new events or actions. For the type of network we wish to analyze, we are motivated to adopt such a model as well: specifically, we assume that at the beginning of each time slot, all links attempt entanglement generation, and that
network state is reset to the empty state (one in which no entangled links are present) at the end of each time slot. Time slot length is assumed to be long enough for all links to be able to perform the necessary number of attempts. For instance, in the absence of multiplexing or with, \eg, spatial multiplexing, time slot length need only be as long as the propagation delay of the longest link in the network, $\tau_{\max}$. Certain temporal multiplexing proposals, on the other hand, might require a time slot length that is a multiple of $\tau_{\max}$. 
Entanglement swapping, namely Bell state measurements (BSMs), is performed at the end of each time slot, before elementary link-level entanglement expires.  This model is descriptive of a quantum network with highly limited quantum memory coherence times -- for an example, networks with atomic ensemble memories \cite{duan2001long} necessitate an almost immediate use of entanglement once it has been generated \cite{maier2020investigating,rabbie2020simulation}. Finally, we assume in our model that each BSM succeeds with an \textit{entanglement swapping success probability} $q\in(0,1]$.

Figure \ref{fig:netEx} illustrates an example quantum network with two nodes, $s$ and $t$, wishing to share as many entangled states as possible. In a heterogeneous network such as this one, each link $(i,j)$ has its own entanglement generation success probability $p_{ij}$, and each network node $n_i$ has its own entanglement swapping success probability $q_i$. Even if the network in Figure \ref{fig:netEx} were non-multiplexed, so that each pair of neighboring nodes shares at most one bipartite entangled state at any time, we may nevertheless already encounter a non-trivial decision-making scenario. Consider, for example, a time instance when all entangled links are present: the network then has the option of connecting $s$ and $t$ via $(a)$ two paths, $s-s_1-s_2-t$ and $s-s_3-s_4-t$, or via $(b)$ a single path, $s-s_3-s_2-t$. Even though option $(a)$ may result in two entangled states while option $(b)$ produces at most one, the optimal choice depends on the swapping success probabilities of the nodes: if, for instance, $q_1$ and $q_4$ are extremely low compared to $q_2$ and $q_3$, then one may opt for $(b)$ to avoid actions that result, with high probability, in no $\st$ entanglement at all.

\begin{figure}[t]
\centering
\includegraphics[]{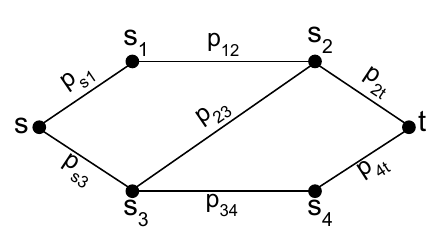}
\caption{Example of a quantum network where two nodes, $s$ and $t$, wish to share entanglement. Link $(i,j)$ generates entanglement with probability $p_{ij}$, and node $n_i$ performs entanglement swapping with success probability $q_i$.}
\label{fig:netEx}
\end{figure}

More concretely, the bipartite entanglement capacity of the above-mentioned network state is given by 
\begin{align*}\max\{q_1q_2+q_3q_4, q_2q_3\},\end{align*}
where the capacity of each path is defined as the product of entanglement swapping probabilities of the nodes on that path. The capacity of a network state $S$ may then be defined as
\begin{align}
C_S \coloneqq \max\limits_{R\in \mathcal{D}}\sum\limits_{r\in R}f(r)\prod\limits_{\substack{i\in r,\\i\neq s,t}}q_i,
\label{eq:netStateCap}
\end{align}
where $\mathcal{D}$ is the set of all possible combinations of link-disjoint $\st$ path sets in the network state, $r\in R$ is an $\st$ path that is link-disjoint with all other paths in $R$, $i\in r$ is a node that lies on path $r$; and $f(r)$ is a function defined on path $r$. $f(r)$ may, in principle, incorporate an entanglement measure (\eg, secret key rate, entanglement entropy, distillable entanglement, \etc) that is most relevant to the application that is being executed by nodes $s$ and $t$. As such entanglement measures typically depend on the specific hardware of the network nodes, link lengths, noise processes, entanglement generation protocols, among other factors (which collectively determine the type of bipartite entanglement that results on each link, as well as how each state evolves with time), we take here $f(r)=1$, for all $\st$ paths $r\in R$ to keep the analysis both simple and general. In other words, we compute only the expected \textit{number} of entangled pairs shared by $s$ and $t$ at the end of swapping. We remark that in the context of entanglement routing, a more general definition of network state capacity $C_S$ may exist: (\ref{eq:netStateCap}) need not be additive, for instance.

Given our definition of network state capacity above, the overall bipartite entanglement capacity of the network (with reference to nodes $s$, $t$) is given by
\begin{align}
C = \sum\limits_{S\in\mathcal{S}}P_SC_S,
\label{eq:capacity}
\end{align}
where $\mathcal{S}$ is the set of all possible network states, and $P_S$ is the probability of observing the network in state $S$ (see more on this in Section \ref{sec:probform}). While in this work we compute $C$ exactly for a number of networks, it is in principle possible to approximate it by computing the expectation over a subset of $\mathcal{S}$ corresponding to the most likely network states (with highest values of $P_S$). Arbitrarily tight upper- and lower-bounds on $C$ can be obtained by increasing the size of this subset, and keeping in mind that the minimum and maximum $C_S$ are given by $0$ and $C_{S_{\text{full}}}$, respectively, where $S_{\text{full}}$ is the network state with all elementary link-level entanglement successfully established. 

We emphasize that the definition of network capacity provided in (\ref{eq:capacity}) is given for the specific setup described earlier in the section, namely for the case in which quantum nodes are only capable of performing BSMs. This formula is a special instance of certain information-theoretic definitions found in literature that consider arbitrary and adaptive local operations and classical communication (LOCC) -- see Section~\ref{sec:relatedwork} for further discussion.

The remainder of this manuscript is organized as follows: in Section \ref{sec:relatedwork}, we discuss relevant background and related work. In Section \ref{sec:probform}, we  state our assumptions, describe our model of a quantum network, and formulate the problem in terms of $\st$ flow in a directed graph. In Section \ref{sec:miqcp}, we introduce the MIQCP for non-multiplexed networks and prove its validity, while in Section \ref{sec:multiplexed} we extend the framework to multiplexed networks. We present numerical examples in Section \ref{sec:numexamples} and make concluding remarks in Section \ref{sec:concl}.

\section{Related Work}
\label{sec:relatedwork}
Request scheduling, path selection, and entanglement routing -- problems which are closely interrelated in quantum networking -- have all previously received attention in literature.
Caleffi studied optimal entanglement routing in \cite{caleffi2017optimal}, where the goal was to determine the best (single) path between two nodes in a quantum network. Van Meter \etal adapted Dijkstra’s algorithm to quantum networks and demonstrated its usefulness for path selection in \cite{van2013path}.
Gyongyosi \etal proposed a decentralized routing scheme for finding the shortest path in a quantum network \cite{gyongyosi2017entanglement}. A number of virtual graph construction and routing techniques have also been proposed, \eg, \cite{schoute2016shortcuts, gyongyosi2018decentralized, chakraborty2019distributed}. Such approaches are especially useful for reducing quantum memory requirements on quantum nodes, as well as lowering latency during the routing process.

As noted in the previous section, a host of studies take an information-theoretic approach in their analyses of quantum network capacity and development of entanglement distribution protocols. Pirandola in \cite{pirandola2019end}, and Azuma \etal in \cite{azuma2016fundamental} establish upper bounds on the maximum rate at which entanglement can be distributed to two nodes in a quantum network. 
The authors of \cite{pirandola2019end} and \cite{harney2022end} considered both single- and multi-path routing scenarios, with the latter focusing on realistic networks, \eg, those with imperfect repeaters. In \cite{azuma2017aggregating}, the authors study entanglement distribution protocols (again from an information-theoretic perspective), for a family of quantum networks wherein independent repeater schemes run on each link. At the link level, details describing specific processes that generate entanglement, including \eg, purification, are abstracted away. The authors use a graph-theoretic result known as Menger’s Theorem, which quantifies the number of edge-disjoint $\st$ paths in a network, to derive protocols that can achieve information-theoretic limits. While the work mainly focuses on deterministic entanglement swapping, the authors provide a brief treatment of the probabilistic-swap case by proposing that swaps can be performed along each $\st$ path in a "knockout tournament" manner. In the context of our problem, neither this technique, nor Menger’s Theorem, can be applied directly: the swapping probabilities play a critical role in determining the optimal set of edge-disjoint paths to be utilized by $s$ and $t$.

In \cite{bauml2020linear},
the authors presented linear programs based on the max-flow min-cut theorem, with the goal of efficiently computing bounds on entanglement distribution rates. We remark that the max-flow min-cut theorem is also
inapplicable to our problem, as it does not offer assurances that flow will not split or merge. A number of studies obtained exact results only in idealized situations, or specific network topologies: for instance,
the authors of \cite{pirandola2019end} analyzed routing with ideal repeater nodes (namely, with perfect BSMs) and reduced the problem to finding the maximum flow through a quantum network. 
Dai \etal proposed a linear program for optimal remote entanglement distribution, in a scenario where quantum memories have the capability to store qubits for an infinite amount of time \cite{dai2020optimal}. Cicconetti \etal optimized request scheduling and path selection in quantum networks with multiple user requests over a finite time horizon in \cite{cicconetti2021request}, but did not consider entanglement swap failures. As mentioned in the previous section, we consider networks which reset the state at the end of each time slot; this means that there is no need to consider a time horizon beyond a single time slot.

Yet another collection of studies opted for heuristic strategies. Le \etal applied a deep neural network to address the problem of servicing multiple end-to-end entanglement requests in a quantum network where entanglement swapping always succeeds \cite{le2022dqra}. Victora \etal optimized distillable entanglement in quantum networks in \cite{victora2020purification}; here, multiple paths were considered, although entanglement swapping was assumed to succeed with unit probability. Several other QoS (Quality of Service)-conscious entanglement routing algorithms have been proposed, \eg, \cite{yang2022online,nguyen2022multiple,zeng2022multi,zhang2022multipath,nguyen2022maximizing,amer2020efficient,zhao2021redundant,ghaderibaneh2022pre,li2022fidelity,li2021effective}, with examples of QoS being bandwidth or throughput, latency, efficient resource use, fidelity, secret key rate; and number of and fairness to concurrent flows. 
Of these studies, perhaps the most similar \textit{formulation} to that of our work is found in \cite{zhao2021redundant}, although the problem is posed with an emphasis on resource allocation: in particular, redundant resource provisioning is proposed before any entanglement generation attempts take place. This work improved that of Shi \etal which introduced another (in general non-optimal) algorithm for concurrent entanglement routing \cite{shi2020concurrent}.

Another study with a similar formulation to that of ours (albeit, with the broader aim of distributing entanglement to multiple simultaneously-communicating source-destination pairs in a network), is that of \cite{chakraborty2020entanglement}, wherein the authors constructed a Linear Program (LP) to solve the entanglement distribution problem using a multi-commodity flow approach. It was assumed that entanglement swapping succeeds probabilistically although all repeaters were assumed to have equal BSM success probabilities. Users were assumed to have minimum end-to-end fidelity requirements which translated into a maximum permitted number of swapping operations. 
The rate maximization problem was solved using an edge-based formulation which involved the construction of an extended graph, resulting in an LP that scales efficiently with the number of variables within the optimization problem. The objective function of the optimization framework then summed over the flow out of all source nodes in the extended graph, but the assumption of equal-probability entanglement swaps allowed the authors to treat same-length paths in this graph (\ie, paths with the same number of hops) equally. Specifically, any path of length $j$ has a gain factor of $q^{j-1}$, where $q$ is the BSM success probability. The reason why we cannot (at least, directly or through simple/obvious modifications) adopt our problem as a special case of this multi-commodity flow approach, is that in our formulation, node heterogeneity (arising from potentially different swapping success probabilities) does not allow us to de-couple gain factors $q_i$ from the flow variables -- two equal-length paths could have substantially different flow capacities, making the $q_i$'s a more fundamental aspect of the problem.

Pant \etal investigated multi-path entanglement routing in a setting where BSMs succeed probabilistically \cite{pant2019routing}. The authors proposed routing protocols based on global or local knowledge of the network topology and link state. For a grid network topology, the authors introduced a greedy algorithm, which while not optimal in general, yields a rate that is within a constant factor of the optimum. As the setup and assumptions in this work are most similar to those in our paper, we use a similar algorithm as a comparison basis of our results. 
Patil \etal improved on the results of \cite{pant2019routing} in \cite{patil2022entanglement}, by allowing {$n$-GHZ} measurements at network nodes, for $n\geq 3$. Leone \etal also experimented with multi-path routing algorithms but without optimality guarantees \cite{leone2021qunet}. 
On a final note,
 entanglement routing has also been studied within the context of entanglement percolation, as in \cite{acin2007entanglement}; 
this, however, lies outside the scope of the current study.

\section{Problem Formulation}
\label{sec:probform}
In this work we consider a quantum network that operates under the following assumptions:
\begin{itemize}
\item[(1)] The network is comprised of quantum repeater nodes interconnected by physical links (\eg, optical fiber) in an arbitrary topology.
\item[(2)] During execution, time is discretized into slots. A time slot is divided into two stages: I, wherein all links attempt entanglement generation, and II, wherein all nodes attempt entanglement swapping. We will sometimes refer to the network state immediately after stage I as a \textit{snapshot}. In stage II, only BSMs may occur -- nodes are incapable of other types of entangling measurements.
\item[(3)] A network node has a number of links, each of which is allocated quantum storage wherein entangled states may be stored. If the link implements multiplexing (\eg, temporal, spatial, or frequency), so that it may generate up to $k>1$ entangled states per time slot, then the link has $k$ dedicated storage qubits; otherwise, it has a single storage qubit.
\item[(4)] At the end of a time slot (after entanglement swapping has taken place), any unconsumed entanglement -- entangled states that were not involved in swapping operations --  is discarded.  
\item[(5)] End-to-end entanglement, when generated successfully across any $\st$ path in the network using the processes described above, is assumed to be of sufficiently high fidelity for the application(s) at nodes $s$ and $t$ desiring the Bell pair.
\end{itemize}

The first part of assumption (2) implies that the network is operating in a coordinated manner: first, entanglement is attempted, and subsequently entanglement routing decisions are executed using global knowledge of the network via entanglement swapping. Such schemes, which may be accomplished with the help of a centralized controller, have been proposed as a means of scaling up quantum networks and ensuring QoS to applications, \eg, delivery of high fidelity end-to-end entanglement. See, \eg, \cite{skrzypczyk2021architecture} which proposes such an architecture as a way of circumventing some of the limitations of near-term multi-user quantum networks, \cite{pompili2021experimental} which adopts the aforementioned, and \cite{kozlowski2019towards} where the authors suggest a software-defined networking (SDN) approach towards such ends.

\section{``Leaky'' Disjoint $\st$ Flow in Graphs with Unit Edge Capacity}
\label{sec:miqcp}
\begin{figure}
\centering
\subfloat[snapshot 1]{\includegraphics[width=0.8\linewidth]{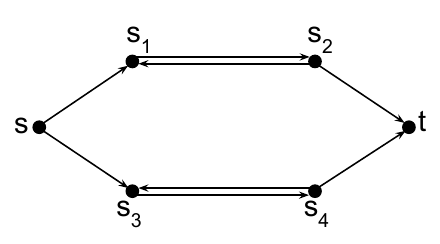}\label{fig:snap1}}\quad
\subfloat[snapshot 2]{\includegraphics[width=0.8\linewidth]{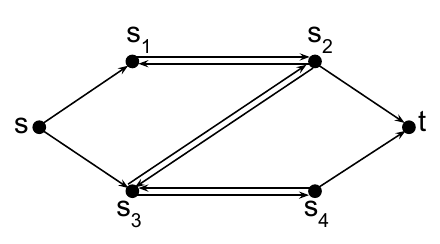}\label{fig:snap2}}
\caption{Two possible network snapshots after link-level entanglement generation attempts.}
\label{fig:netex}
\end{figure}
We focus now on solving the flow problem for a non-multiplexed network snapshot. While entanglement distribution is inherently symmetric (with no difference in treatment of the end nodes), when cast as a flow problem necessity arises for one node to be designated the source ($s$) and the other the terminal ($t$): flow \textit{has} directionality. Either node may assume either role -- see discussion at the end of this section. To proceed with this formulation, however, we must devise a directed representation of the undirected graph $\mathcal{G}$ representing a network state (snapshot) immediately after stage I (link-level entanglement generation). 

 From now on, we may assume that the snapshot $\mathcal{G}$ contains at least one path between $s$ and $t$ (since otherwise, the capacity of the snapshot is trivially zero). Let $G=(V,E)$ represent the corresponding directed graph for the snapshot. We construct this graph by converting each undirected edge $(u,v)$ of $\mathcal{G}$ into a set of directed edges $(u,v)$ and $(v,u)$, for any $u,v\notin\{s,t\}$. For undirected edges $(s,u)$, it suffices to include only directed edges $(s,u)$ within $E$, and for any undirected edges $(t,u)$, it suffices to include only directed edges $(u,t)$ within $E$. Figure \ref{fig:netex} depicts two possible directed snapshots resulting from entanglement generation attempt stages of the network shown in Figure \ref{fig:netEx}.

In addition to the basic flow conservation constraint frequently encountered in flow problems, our formulation must also account for the following two phenomena: $(i)$ entanglement swapping is subject to failures, and $(ii)$ each swap operation consumes two entangled links. Incorporating $(i)$ can be accomplished simply by modeling each node in the snapshot as lossy, with a gain factor equal to the entanglement swapping success probability. Incorporating $(ii)$ on the other hand is less straightforward, as it involves designing constraints that ensure:
\begin{itemize}
\item two opposing flows cannot coexist, \ie, edges $(i,j)$ and $(j,i)$ cannot simultaneously carry flow -- this corresponds to the assumption that any two neighboring nodes have at most one entangled link between them;
\item flows do not split into two or more portions as they propagate through the network; and
\item each flow remains disjoint from other flows (flows do not merge); in the special case of a network with unit edge capacities, this amounts to ensuring that each edge is used by at most one flow.
\end{itemize}
We call the second constraint \textit{flow intactness} and the third \textit{flow disjointness}.
As we shall see, these requirements may be satisfied with the help of indicator variables that help direct and restrict the flow as necessary.
We begin by introducing a mixed-integer quadratically constrained program (MIQCP) that solves this constrained flow problem in Section \ref{subsec:miqcp_form}, where we also intuitively explain the purpose of each constraint. In Section \ref{subsec:miqcp_val}, we prove more rigorously the validity of the MIQCP.
\subsection{MIQCP Formulation}
\label{subsec:miqcp_form}
Let $F_{ij}$, $i\neq t$, $j\neq s$, represent flow from node $i$ to $j$.
Consider the following objective function and set of constraints:
\begin{align}
&\max\sum\limits_{(j,t)\in E} F_{jt}\label{eq:obj}\\
&0\leq F_{ij}\leq 1, \quad\forall (i,j)\in E,\\
&x_{ijk}\in \{0,1\},\quad \forall (i,j),~(j,k)\in E ~/\ i\neq k,~ j\neq s,t,\\
&\sum\limits_{k}(x_{ijk}+x_{kji})\leq 1\quad\forall(i,j)\in E,\label{eq:xijk3}\\
&x_{ijk}(F_{ij}q_j - F_{jk}) = 0,\quad \forall (i,j), (j,k) \in E,\label{eq:xflow1}\\
&\sum\limits_k x_{ijk} \geq F_{ij}, \quad\forall (i,j) \in E~/ j \neq t,\label{eq:xflow2}\\
&\sum\limits_{j\in V}(F_{ij}-F_{ji}q_i)=0,\quad\forall i\in V,~ i\neq s,t.\label{eq:flowcons}
\end{align}

An intuitive explanation of each line is as follows: first,
the objective function maximizes the flow to node $t$, and since flow can only originate at node $s$, the solution maximizes the overall $\st$ flow. In Section \ref{sec:miqcp_validity}, we explain why this is equivalent to obtaining the capacity of our snapshot graph.
Next, we have the constraint on the flow variables $F_{ij}$: the upper bound of one is due to the fact that we have one quantum memory per physical link interface, and $F_{ij}\in[0,1]$ is due to the fact that our "internal" network nodes are lossy (internal nodes are all nodes in $V$ except $s$ and $t$). 
In the next line, we introduce binary variables $x_{ijk}$ defined for an internal node $j$ and two of its physical neighbors $i$ and $k$. These variables
help us perform a constrained form of bipartite matching between the incoming and the outgoing edges to/from node $j$ . 

Then, constraint (\ref{eq:xijk3}) enacts the flow disjointness and intactness constraints, which consist of ensuring that an edge admits flow from at most one other node, and that an edge releases flow to at most one other node. 
Further, constraint (\ref{eq:xijk3}) ensures that two neighboring nodes $i$ and $j$ are only allowed to share flow in one direction.
Next, constraint (\ref{eq:xflow1}) ensures that whenever flow travels from $i$ to $j$ and into $k$, it must be conserved. Note that this is a more refined flow constraint than that of $(\ref{eq:flowcons})$, which conserves flow at each node but does not necessarily ensure the proper amount of flow on each edge as dictated by the intactness and disjointness constraints. Finally, constraint (\ref{eq:xflow2}) ensures that flow cannot exist on an edge unless a bipartite matching constraint allows it -- \ie, for a flow to be non-zero, there must be a variable $x_{ijk}$ indicating that the edge is actually being used to carry flow.
\subsection{MIQCP Validity}
\label{sec:miqcp_validity}
\begin{figure}[t]
\centering
\subfloat[]{\includegraphics[width=0.2\textwidth]{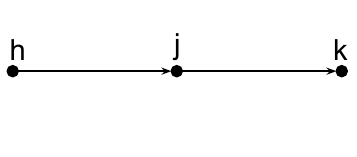}\label{fig:EX1}}\qquad
\subfloat[]{\includegraphics[width=0.2\textwidth]{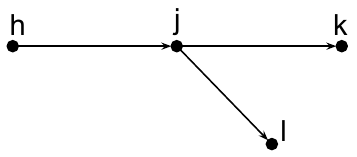}\label{fig:EX2}}\qquad
\label{fig:lemmaproof1}
\caption{Different scenarios relevant to Lemma \ref{lemma:noxijknoflow} proof. Directed arrows indicate presence of flow in that direction.}
\end{figure}
\label{subsec:miqcp_val}
We next study the validity of the MIQCP given by $(\ref{eq:obj})$-$(\ref{eq:flowcons})$. This involves proving the following:
\begin{itemize}
\item That neighboring nodes share flow in one direction at a time: 
\begin{align}
F_{ij}>0 \implies F_{ji}=0,~ \forall (i,j)\in E;
\end{align}
\item That flow intactness and disjointness hold. Formally, this means the following: given a node $j\neq t$, first construct the sets
\begin{align}
    \mathcal{I}_j &= \{F_{ij}: (i,j)\in E, F_{ij}>0\},\\
    \mathcal{O}_j &= \{F_{jk}: (j,k)\in E, F_{jk}>0\}.
\end{align}
In other words, $\mathcal{I}_j$ is the set of non-zero incoming flows and $\mathcal{O}_j$ is the set of outgoing flows with respect to $j$. We require the existence of a one-to-one correspondence between these sets, which satisfies the condition that for a flow $F_{ij}\in \mathcal{I}_j$, 
$F_{ij}q_j=F_{jk}$ for a flow $F_{jk}\in\mathcal{O}_j$;
\item That flows in the network follow valid $\st$ paths. A valid, $l$-hop path is defined by a sequence of positive flows $F_{i_0,i_1}$, $F_{i_1,i_2}$, $\dots$, $F_{i_l,i_{l+1}}$ satisfying $F_{i_k,i_{k+1}}q_{i_{k+1}}=F_{i_{k+1}}$, $k=1,2,\dots,l$, where the set of nodes $\{i_1,i_2,\dots, i_l\}\subseteq V$. Note that nodes in the path may in principle repeat, although this would not yield optimal results with sub-unit gain factors. A valid $\st$ path has the additional property that $i_0=s$ and $i_{l+1}=t$;
\item That each constraint is necessary.
\end{itemize}
Having defined path validity, as well as flow intactness and disjointness, we are now able to establish a direct connection between the snapshot capacity (\ref{eq:netStateCap}) with $f(r)=1$, $\forall r\in R$, and the MIQCP objective (\ref{eq:obj}). Namely, let us assume that all three conditions hold for any flow in the graph. Then, for a given $F_{jt}>0$, the flow must have originated at node $s$ and followed a valid path $s\to n_1\to n_2\to\dots\to n_k\to t$ through the network, all the while staying intact and disjoint. Since the objective is to maximize flow -- a value that lies in $[0,1]$, it follows that $F_{s,n_1}=1$. Combining this with the aforementioned flow properties allows us to deduce that $F_{jt}=q_1q_2\dots q_k$, where $q_i$, $i\in\{1,\dots,k\}$, is the gain factor associated with node $n_i$. This expression corresponds to a single term in the sum of (\ref{eq:netStateCap}), and the flow maximization objective then results in finding the optimal set of link-disjoint paths in the graph. We thus conclude that the value of the MIQCP objective function for a given snapshot is equivalent to the snapshot's capacity.

\subsubsection{Neighboring nodes share flow in at most one direction}
We begin with some useful lemmas.
\begin{lemma} 
\label{lemma:noxijknoflow}
For any edge $(j,k)$ s.t. $j\neq s$:
\begin{align*}
\sum\limits_{i}x_{ijk}=0 \implies F_{jk} =0.
\end{align*}
Intuitively, this result states that if no edge $(i,j)$ carries flow into $(j,k)$, then the flow value on $(j,k)$ is zero.
\end{lemma}
\begin{proof}
Our proof is by contradiction. Assume that $F_{jk}>0$, for a node $j\neq s$, while $\sum_{i}x_{ijk}=0$. Recall that there are no outgoing edges at node $t$ by our definition of $G$, so we know that $j$ is an internal node. Then by condition (\ref{eq:flowcons}) on node $j$, we see that
\begin{align}
F_{jk}+\sum\limits_{i\neq k}F_{ji}-q_j\sum\limits_{i}F_{ij} = 0,
\label{eq:jflow1}
\end{align}
which implies the existence of some node $h$ such that ${F_{hj}>0}$ (since if no such node exists, (\ref{eq:flowcons}) is violated due to the left-hand side being strictly positive). Thus, we may rewrite (\ref{eq:jflow1}) as
\begin{align}
F_{jk}-q_jF_{hj}+\sum\limits_{i\neq k}F_{ji}-q_j\sum\limits_{i\neq h}F_{ij} = 0.
\label{eq:jflow2}
\end{align}
We do not require that $h\neq k$ for our proof.
Figure \ref{fig:EX1} depicts the scenario where $h \neq k$.

Since $F_{hj}>0$, by condition (\ref{eq:xflow2}) there exists a node $l$ such that $x_{hjl}=1$, and by condition (\ref{eq:xijk3}) applied to edge $(h,j)$, this $l$ is unique. Further, given the way that the $x_{ijk}$ variables are defined, we can be sure that $l\neq h$. 
Thus, if $h=k$, \ie, ${x_{kjl}=1}$, then we know by the definition of the variables $x_{ijk}$ that $l\neq k$.
On the other hand, if $h\neq k$, then by our assumption that $\sum_i x_{ijk}=0$, it follows that $x_{hjk}=0$, so that $l\neq k$.  \textit{I.e.}, in either case, $l\neq k$ -- an important observation that, as we shall see shortly, means that flow $F_{jk}$ will remain in the equation.

By condition (\ref{eq:xflow1}) applied to edges $(h,j)$ and $(j,l)$, we must have
\begin{align}
F_{hj}q_j = F_{jl},
\label{eq:jflow3}
\end{align}
which implies that $F_{jl}$ is strictly positive, so that we have flow from $j$ to $l$, as depicted in Figure \ref{fig:EX2} for the case where $h\neq k$. Thus, we may write (\ref{eq:jflow2}) as
\begin{align}
F_{jk}-q_jF_{hj}+F_{jl}+\sum\limits_{i\neq k,l}F_{ji}-q_j\sum\limits_{i\neq h}F_{ij} = 0,
\end{align}
and by (\ref{eq:jflow3}),
\begin{align}
F_{jk}+\sum\limits_{i\neq k,l}F_{ji}-q_j\sum\limits_{i\neq h}F_{ij} = 0.
\label{eq:jflow4}
\end{align}

Comparing Eqs. (\ref{eq:jflow1}) and (\ref{eq:jflow4}), we see that the only difference is that in the latter, two components on the left-hand side have cancelled each other out: one from each of the sums, while the $F_{jk}$ component has not changed at all. If we now attempt to balance out $F_{jk}$ using another (negative) component from the second sum above, say $F_{gj}$ ($g\neq h)$, then conditions (\ref{eq:xflow2}) and (\ref{eq:xijk3}) ensure the existence of a unique $m\neq g$ for which $x_{gjm}=1$. Note that $m\neq k$ since either $(i)$ $g\neq k$, so that if $m=k$ we would have $x_{gjk}=1$, which contradicts our assumption that $x_{ijk}=0$, $\forall i$, or $(ii)$ $g=k$, so that $x_{kjm}$ does not allow $m=k$ by definition. Further, $m\neq l$, since we already have $x_{hjl}=1$, so we cannot also have $x_{gjl}=1$  (this is prohibited by condition (\ref{eq:xijk3}) applied to edge $(j,l)$). By condition (\ref{eq:xflow1}) we thus have positive flow $F_{jm}$, for $m\neq k,l$, which must come from the first sum in (\ref{eq:jflow4}). 

It is evident that any time we attempt to balance out flow $F_{jk}$ with a negative component from (\ref{eq:jflow1}), we force another (strictly positive) flow into node $k$, and this positive flow cancels the negative component, leaving $F_{jk}$ unaffected. We can continue this cancellation process until there are no longer any negative components left in the flow balance equation, which would leave us with a strictly positive value on the left-hand side that is $\geq F_{jk}$. From this, we conclude that $\sum_{i}x_{ijk}=0 \implies F_{jk} =0$.
\end{proof}

\begin{lemma}
\label{lemma:nobidirectflow}
For any edge $(i,j) \in E$, $F_{ij}>0\implies F_{ji}=0$.
\end{lemma}
\begin{proof}
Note that by $G$'s construction, edges $(j,s)$ and $(t,j)$ do not exist, so we are only concerned with edges $(i,j)$ for which $i,j\notin\{s,t\}$.
If $F_{ij}>0$, then by condition (\ref{eq:xflow2}), there exists a node $\pr{k}$ for which $x_{ij\pr{k}}=1$. Then by condition (\ref{eq:xijk3}), $\sum_{k}x_{kji}=0$. By Lemma \ref{lemma:noxijknoflow}, ${F_{ji}=0}$.
\end{proof}

\subsubsection{Flow intactness and disjointness}
\label{subsubsec:flow_intact_disjoint}
\begin{lemma}
\label{lemma:flowsplitmerge}
Flow cannot split or merge. In other words, there is a one-to-one correspondence between incoming and outgoing flow for any node's flow balance equation. A consequence of this is that there is an even number of components for any flow conservation equation (\ref{eq:flowcons}).
\end{lemma}
\begin{proof}
For any node $i\neq s,t$, the flow conservation constraint (or flow balance equation), (\ref{eq:flowcons}), can be rewritten as
\begin{align}
\sum\limits_{\pr{j}\in V/F_{i\pr{j}}>0}F_{i\pr{j}}=q_i\sum\limits_{j\in V/F_{ji}>0}F_{ji},
\label{eq:flowconsi}
\end{align}
where, by Lemma \ref{lemma:nobidirectflow}, no $\pr{j}$ on the left-hand side may equal to a $j$ on the right-hand side. Consider any $F_{ji}$. Since it is positive (and since $i\neq t$), we must have, by constraint (\ref{eq:xflow2}), some node $k$ for which $x_{jik}=1$, and by constraint (\ref{eq:xijk3}) this $k$ must be unique. Then by constraint (\ref{eq:xflow1}), we must have $F_{ji}q_i = F_{ik}$, implying that $F_{ik}>0$. Then, Lemma \ref{lemma:noxijknoflow} implies that there exists some node $l$ for which $x_{lik}=1$, and constraint (\ref{eq:xijk3}) ensures that this $l$ is unique. But since we already have $x_{jik}=1$, it must be that $l=j$. Thus, for every flow $F_{ji}$ on the right side of (\ref{eq:flowconsi}) there is a unique corresponding flow $F_{ik}$ on the left side of (\ref{eq:flowconsi}), proving the lemma.
\end{proof}
Lemma \ref{lemma:flowsplitmerge} implies flow intactness and disjointness. 

\subsubsection{Path validity}
\label{sec:miqcp:path_val}
\begin{figure}
\centering
\includegraphics[width=0.48\linewidth]{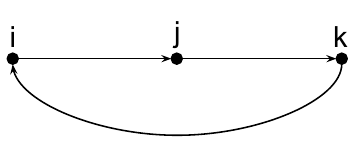}\label{fig:ex3}
\caption{A valid flow scenario described in Section \ref{sec:miqcp:path_val}.}
\label{fig:ex_flow_cycle}
\end{figure}
First, we show that flows in the network originate at node $s$ and terminate at node $t$. To show this, we consider a flow $F_{ij}>0$, with $i,j\neq s,t$, respectively. Let us first see where the flow ends: since $j\neq t$, there exists a node $k\neq i$ for which $x_{ijk}=1$ (cond. (\ref{eq:xflow2})), which then implies a positive $F_{jk}$ (cond. (\ref{eq:xflow1})).  If $k=t$, we are done, as cond. (\ref{eq:xflow2}) no longer applies, and there are no outgoing edges from $t$. Otherwise, cond. (\ref{eq:xflow2}) is applicable: $\exists~ l\neq j$ s.t. $x_{jkl}=1$, implying a positive flow $F_{kl}$. We note here that in principle, $l$ could be the same node as $i$ as shown in Figure \ref{fig:ex_flow_cycle}, albeit likely a suboptimal option, especially if $i$ is a lossy node (\ie, where $q_i<1$). Continued applications of conds. (\ref{eq:xflow2}) and (\ref{eq:xflow1}), subject to constraint (\ref{eq:xijk3}), continue until the sole stopping condition -- positive flow that ends at node $t$.

We now focus on the origin of flow $F_{ij}$, by backtracking it from node $i$ with repeated applications of Lemma \ref{lemma:noxijknoflow}. First, note that by the contrapositive of the lemma's hypothesis, $F_{ij}>0\implies \sum_{l}x_{lij}\neq 0$, implying that $x_{hij}=1$ for some $h\neq j$. By cond. (\ref{eq:xflow1}), $F_{hi}>0$, and if
$h=s$, we are done since  Lemma \ref{lemma:noxijknoflow} no longer applies, and there are no edges that go into $s$. Otherwise, the lemma does apply, and we continue its application along with cond. (\ref{eq:xflow1}) until the only stopping condition is reached, \ie, until we find a flow originating at node $s$. Finally, recall that, by Lemma \ref{lemma:flowsplitmerge}, no flow in the network may split or merge -- together with the discussion above, it is evident that the paths are valid, \ie, they are disjoint, intact paths that originate at $s$ and end at $t$.

\subsubsection{Necessity of each constraint}
\begin{figure}
\centering
\subfloat{\includegraphics[width=0.8\linewidth]{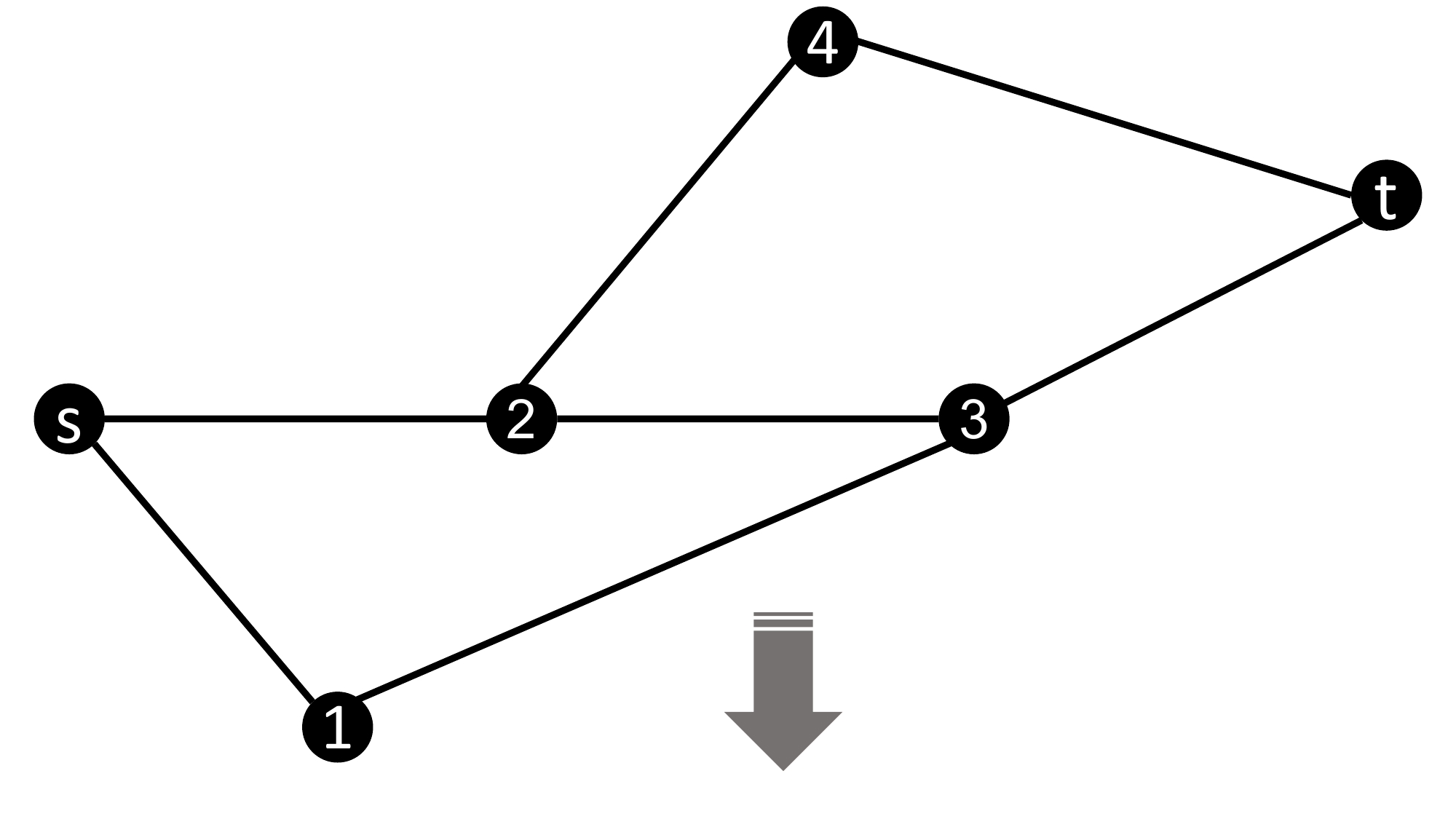}}\\
\subfloat{\includegraphics[width=0.8\linewidth]{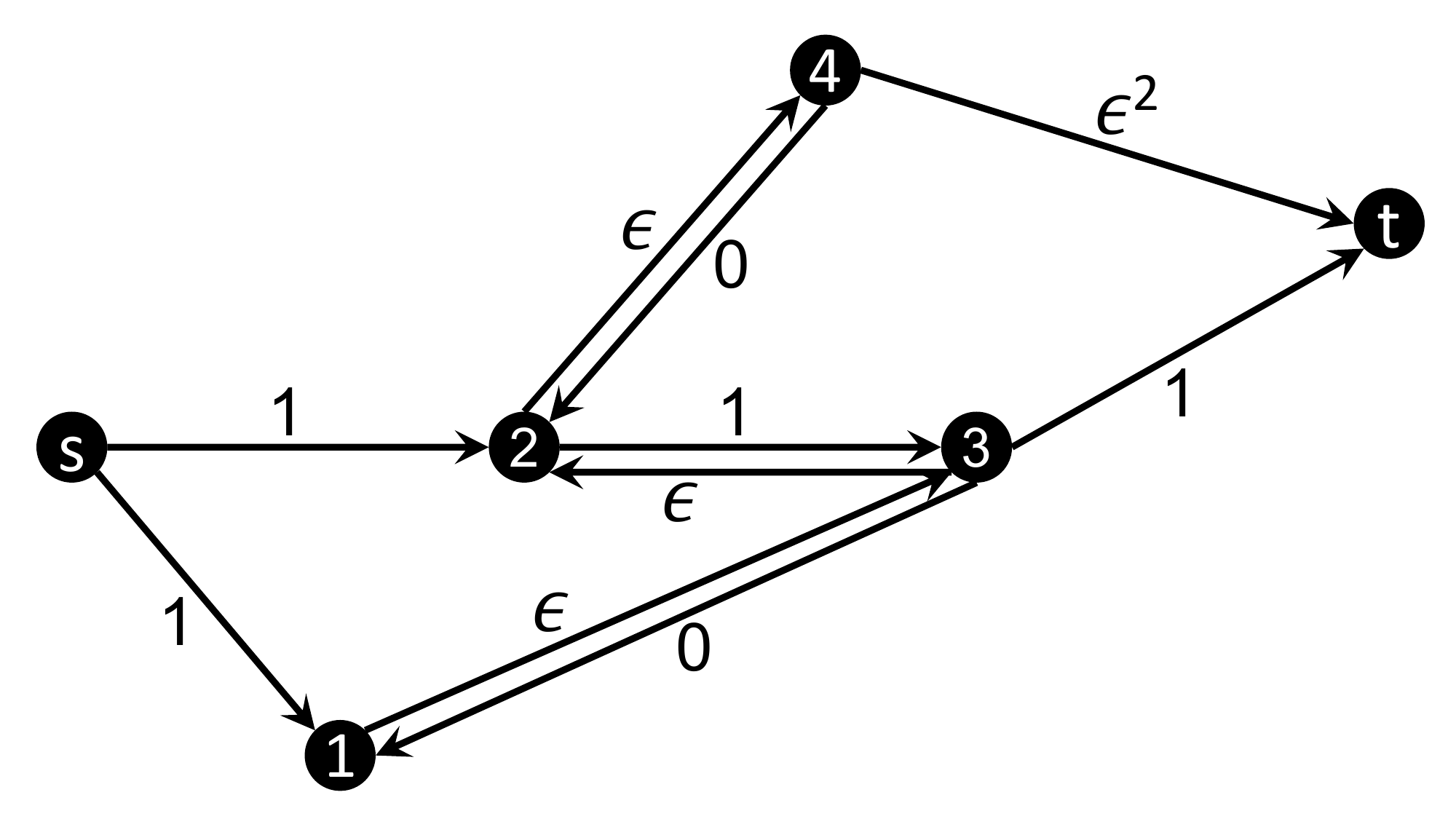}}
\caption{An example illustrating the necessity of constraint (\ref{eq:xijk3}). Nodes $1$ and $4$ have gain factors $\epsilon$, while nodes $2$ and $3$ have unit gain factors.}
\label{fig:constr_xijk3_ex}
\end{figure}
We consider here the nontrivial constraints (\ref{eq:xijk3})-(\ref{eq:flowcons}). While they were all previously helpful in arguing certain flow properties, we have yet to show that each of these constraints is \textit{necessary}. To do so, we consider the MIQCP of Section \ref{subsec:miqcp_form} without each constraint individually, and provide counterexamples that illustrate undesirable flow behavior. 
\begin{table}[t]
\centering
\begin{tabular}{l| l}
Flow on $(i,j)$ & Corresp. $x_{ijk}$\\
\hline
$F_{s1}=1$ & $x_{s13}=1$\\
$F_{13}=\epsilon$ & $x_{132}=1$\\
$F_{32}=\epsilon$ & $x_{324}=1$\\
$F_{24}=\epsilon$ & $x_{24t}=1$\\
$F_{4t}=\epsilon^2$ & N/A\\
$F_{s2}=1$ & $x_{s23}=1$\\
$F_{23}=1$ & $x_{23t}=1$\\
$F_{3t}=1$ & N/A
\end{tabular}
\caption{Variables and their assigned values for the network in Figure \ref{fig:constr_xijk3_ex} (variables assigned 0 not shown).}
\label{tab:constr_xijk3_ex}
\end{table}
First, suppose that (\ref{eq:xijk3}) is removed from the MIQCP. Figure~\ref{fig:constr_xijk3_ex} presents an example where the absence of this constraint allows two opposing flows to coexist -- edge $(2,3)$ carries unit flow, while simultaneously edge $(3,2)$ carries an $\epsilon$ amount of flow; we explain how this happened shortly. In this network, nodes $2$ and $3$ have unit gain factors while nodes $1$ and $4$ have gain factors of $\epsilon$. From the undirected snapshot (top), it is evident that depending on the value of $\epsilon$, it may be optimal to either use the two completely disjoint paths $s-1-3-t$ and $s-2-4-t$, or simply the path $s-2-3-t$. The former yields a snapshot capacity of $2\epsilon$, while the latter yields $1$. Without constraint (\ref{eq:xijk3}), however, it is also possible to obtain a snapshot capacity of $1+\epsilon^2$, if one uses paths $s-2-3-t$ and $s-1-3-2-4-t$. Table~\ref{tab:constr_xijk3_ex} presents all non-zero variables and their assigned values for this network. Note that constraints (\ref{eq:xflow1})-(\ref{eq:flowcons}) are satisfied, but without (\ref{eq:xijk3}) path $s-1-3-2-4-t$, which is not disjoint with path $s-2-3-t$, is active. Constraint (\ref{eq:xijk3}) would have ensured that $x_{s23}$ and $x_{324}$ are not both set to one (similarly for $x_{23t}$ and $x_{132}$).

\begin{figure}
\centering
\subfloat{\includegraphics[width=0.8\linewidth,trim={0 1cm 0 1.4cm},clip]{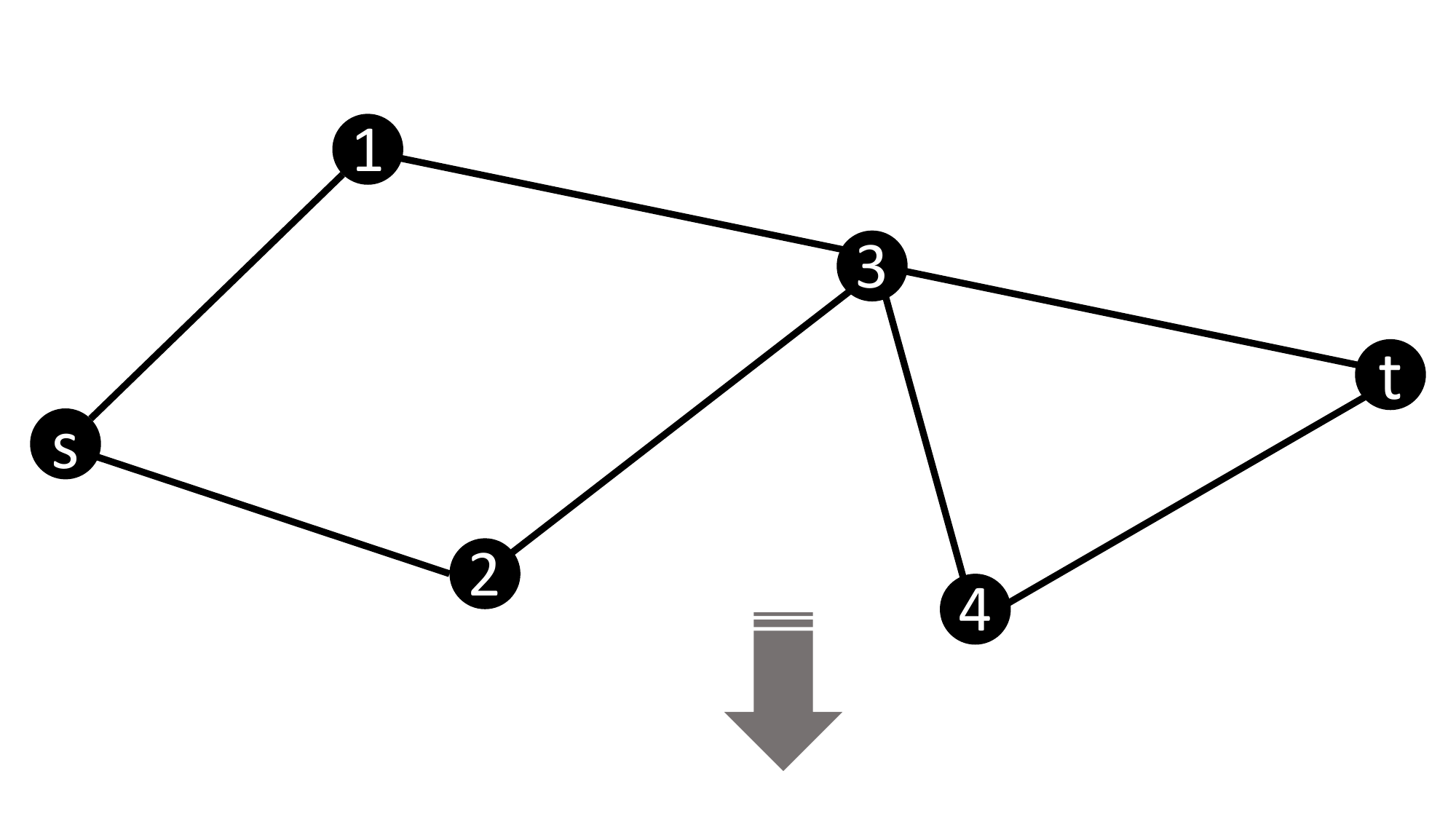}}\\
\subfloat{\includegraphics[width=0.8\linewidth,trim={0 2.5cm 0 1.5cm},clip]{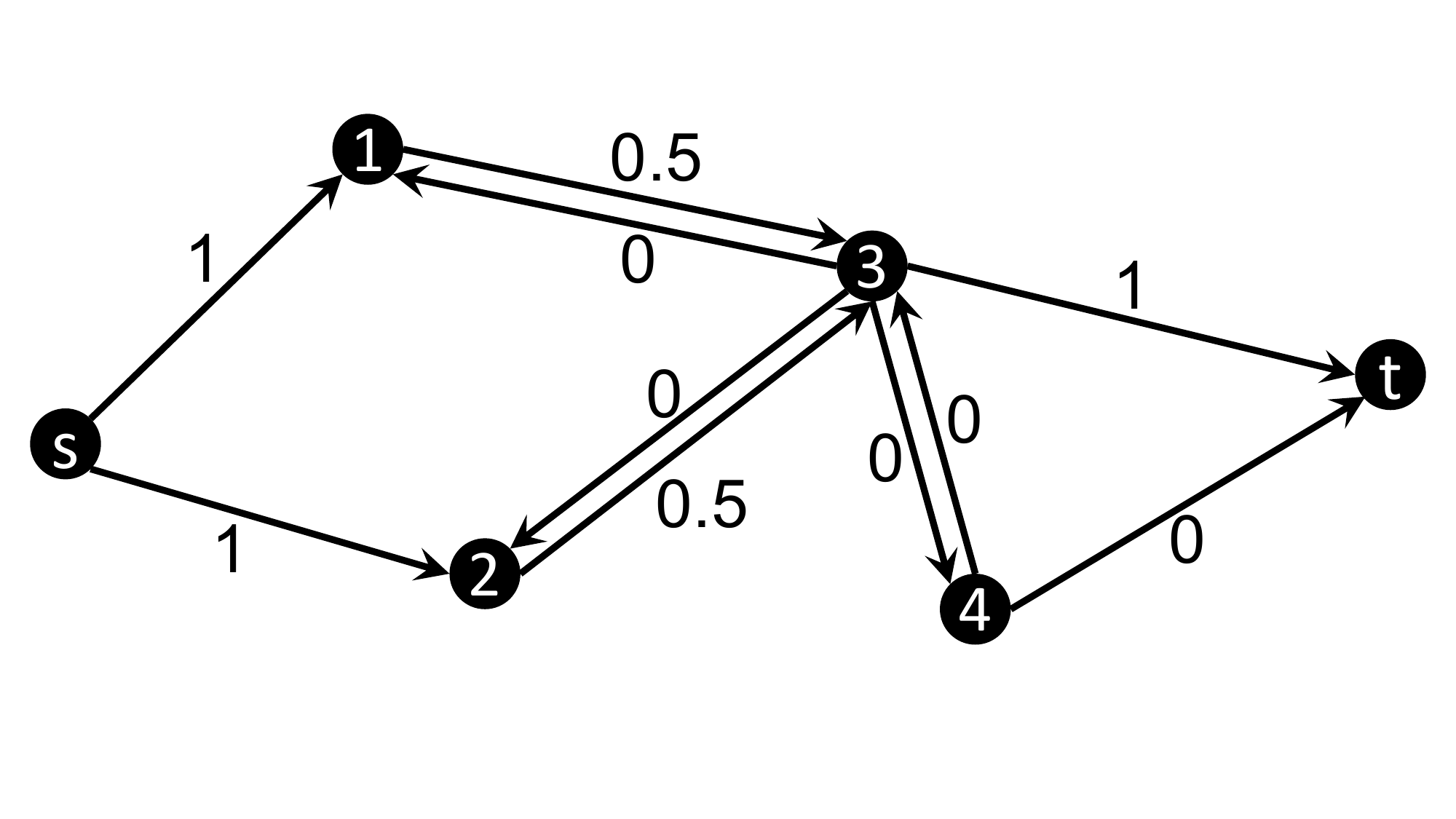}}
\caption{An example illustrating the necessity of constraint (\ref{eq:xflow1}). Nodes $1$ and $2$ have gain factors of $1/2$, node $3$ a unit gain factor, and node $4$ a gain factor of $\epsilon\ll 1$.}
\label{fig:constr_xflow1_ex}
\end{figure}
\begin{table}[t]
\centering
\begin{tabular}{l| l}
Flow on $(i,j)$ & Corresp. $x_{ijk}$\\
\hline
$F_{s1}=1$ & $x_{s13}=1$\\
$F_{13}=0.5$ & $x_{13t}=1$\\
$F_{3t}=1$ & N/A\\
$F_{s2}=1$ & $x_{s23}=1$\\
$F_{23}=0.5$ & $x_{234}=1$
\end{tabular}
\caption{Variables and their assigned values for the network in Figure \ref{fig:constr_xflow1_ex} (variables assigned 0 not shown).}
\label{tab:constr_xflow1_ex}
\end{table}
We next move on to constraint (\ref{eq:xflow1}): to see why it is necessary, consider the network in Figure~\ref{fig:constr_xflow1_ex}. Here, nodes $1$ and $2$ have gain factors of $1/2$, node $3$ a unit gain factor, and node $4$ a gain factor of $\epsilon\ll 1$. The true optimal value of the objective is $0.5+0.5\epsilon<1$, using for example the two edge-disjoint paths $s-1-3-t$ and $s-2-3-4-t$. Without constraint (\ref{eq:xflow1}), however, it is possible to route flow so that the optimal value is $1$ as shown in Figure~\ref{fig:constr_xflow1_ex}, bottom. First, note that with these flow assignments, flow at each internal node is balanced (constr. (\ref{eq:flowcons})). Table~\ref{tab:constr_xflow1_ex} makes it easy to see that (\ref{eq:xflow2}) is also satisfied; the second column of the table can also be used to verify that conditions (\ref{eq:xijk3}) are met. If (\ref{eq:xflow1}) were included, it would have ensured that flows do not merge at node 3.

\begin{figure}
\centering
\subfloat{\includegraphics[width=0.8\linewidth,trim={0 2.3cm 0 1.5cm},clip]{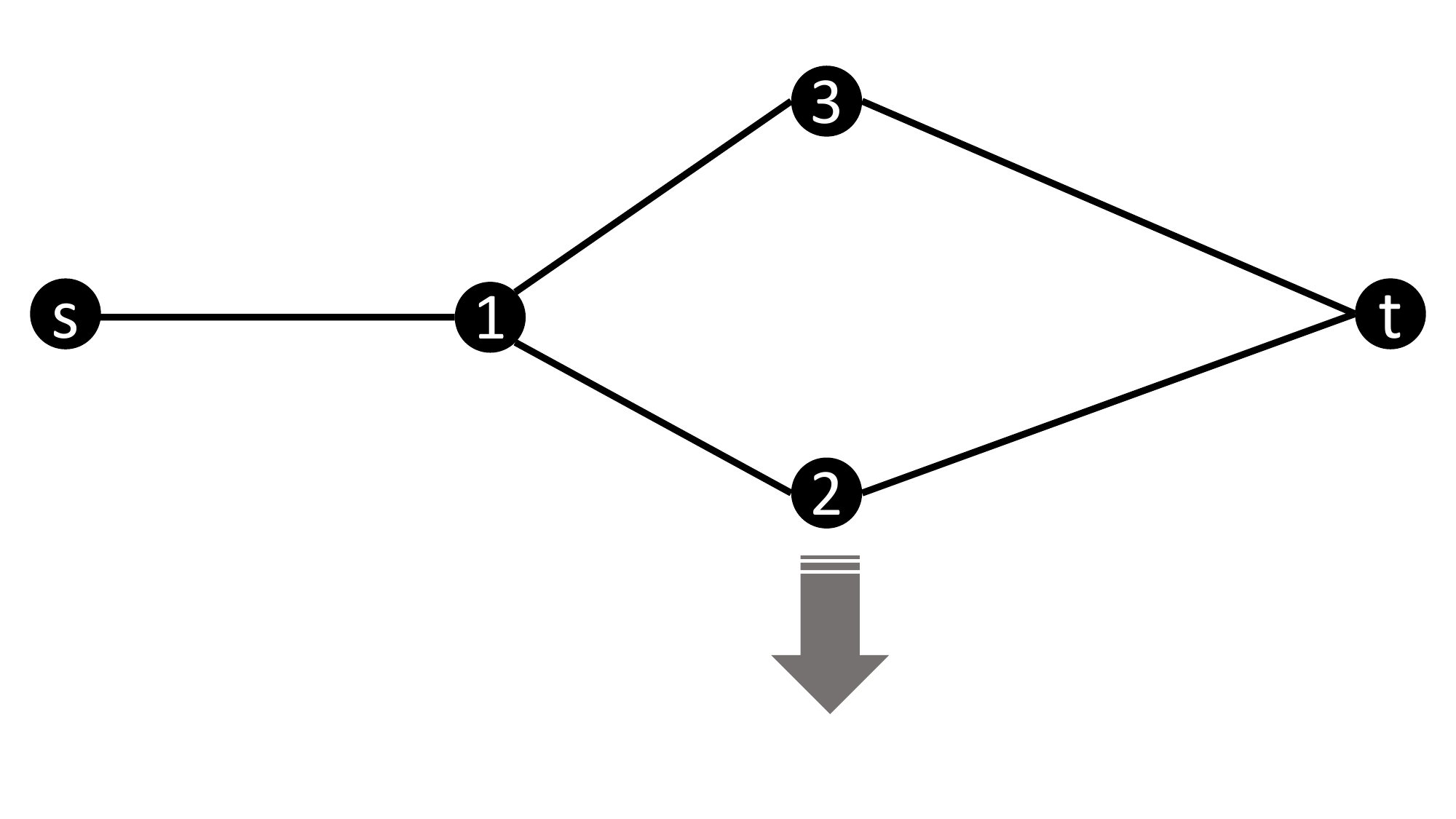}}\\
\subfloat{\includegraphics[width=0.8\linewidth,trim={0 3cm 0 4.1cm},clip]{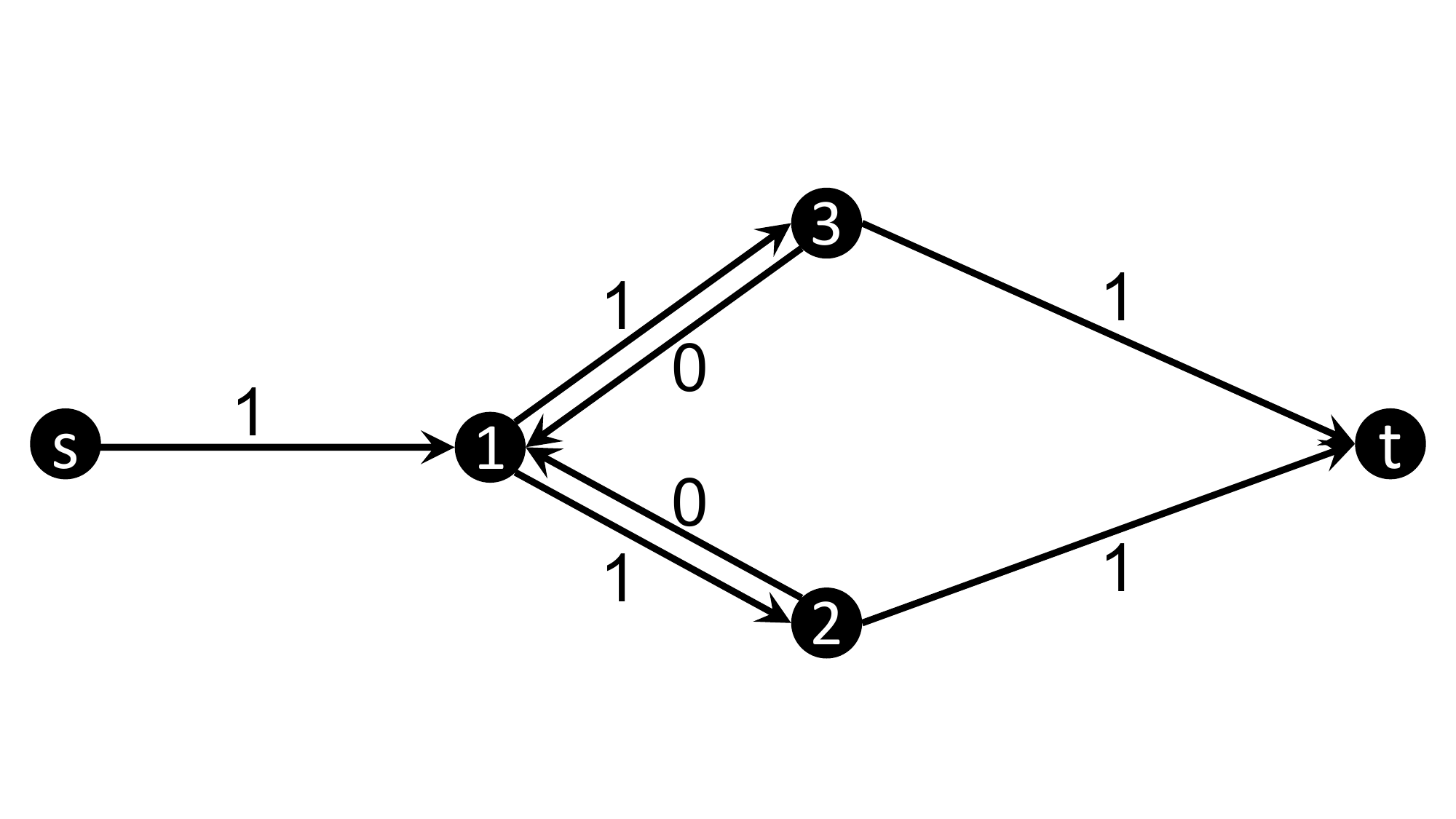}}
\caption{An example illustrating the necessity of constraint (\ref{eq:flowcons}). All internal nodes have unit gain factors.}
\label{fig:constr_flowcons_ex}
\end{figure}
\begin{table}[t]
\centering
\begin{tabular}{l| l}
Flow on $(i,j)$ & Corresp. $x_{ijk}$\\
\hline
$F_{s1}=1$ & $x_{s13}=1$\\
$F_{13}=1$ & $x_{13t}=1$\\
$F_{3t}=1$ & N/A\\
$F_{12}=1$ & $x_{12t}=1$\\
$F_{2t}=1$ & N/A
\end{tabular}
\caption{Variables and their assigned values for the network in Figure \ref{fig:constr_flowcons_ex} (variables assigned 0 not shown).}
\label{tab:constr_flowcons_ex}
\end{table}
Finally, we address constraint (\ref{eq:flowcons}). Figure~\ref{fig:constr_flowcons_ex} presents a network snapshot (top) and corresponding flow formulation (bottom) where the sink node $t$ receives more flow than is released from the source $s$. In this network, all internal nodes have unit gain factors. Table~\ref{tab:constr_flowcons_ex} presents the non-zero flows and their corresponding $x_{ijk}$'s (each set to one, to ensure that constraints (\ref{eq:xflow2}) are satisfied). It can be checked from the right column of the table that constraints (\ref{eq:xijk3}) are also satisfied. Since $x_{s13}$, $x_{13t}$, and $x_{12t}$ are set to one, we have $F_{s1}=F_{13}$, $F_{13}=F_{3t}$, and $F_{12}=F_{2t}$, respectively, to ensure that (\ref{eq:xflow1}) are satisfied. Without constraint (\ref{eq:flowcons}), however, nothing enforces overall flow conservation at node $1$.

We conclude this section with a final remark: that switching the roles of $s$ and $t$ so that $s$ is the sink and $t$ the source, \ie, reversing the flow within the network, yields the same solution to the snapshot flow problem. To see this, recall that flows follow edge-disjoint $\st$ paths, with the additional requirement that if positive flow exists on edge $(i,j)$, then edge $(j,i)$ carries zero flow. This means that to obtain the same solution on a reversed network, one simply reverses each of the $\st$ paths obtained on the original network. This reversal preserves all necessary requirements on the flow, and the amount of flow that reaches node $s$ in the modified network is the same as the amount that would have reached $t$ in the original network (the multiplication order of $q_i$'s is inconsequential). Obtaining a better solution on the reversed network is impossible, as a consequence of the same argument as above -- one could obtain the same solution on the original network using path reversal.

\section{Extension to Multiplexed Links}
\label{sec:multiplexed}
\begin{figure}
\centering
\subfloat[initial snapshot]{\includegraphics[width=0.8\linewidth,trim={0 2.8cm 0 2.2cm},clip]{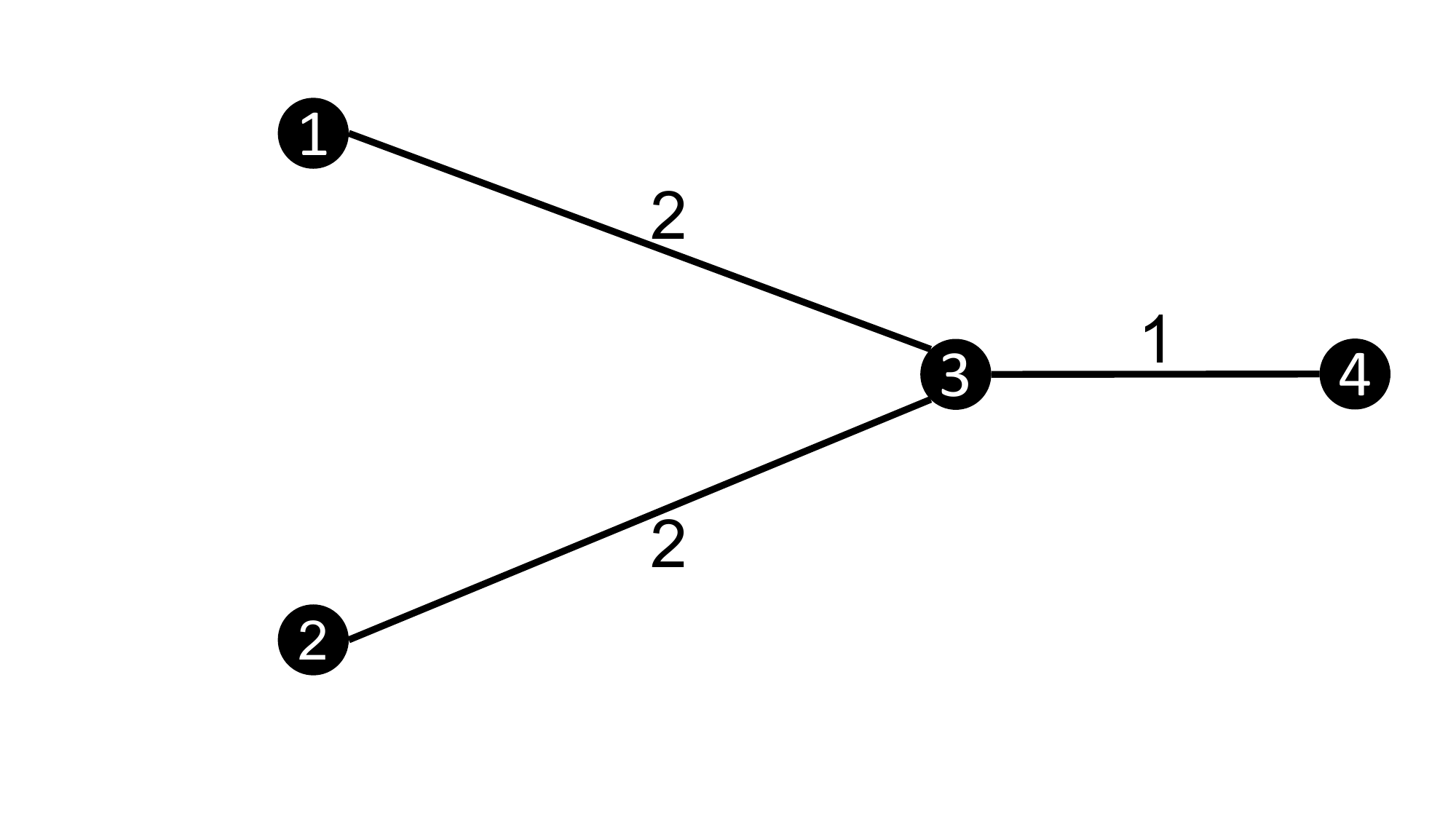}\label{fig:multipl:a}}\\
\subfloat[entangled link transformation]{\includegraphics[width=0.8\linewidth,trim={0 2.8cm 0 2.2cm},clip]{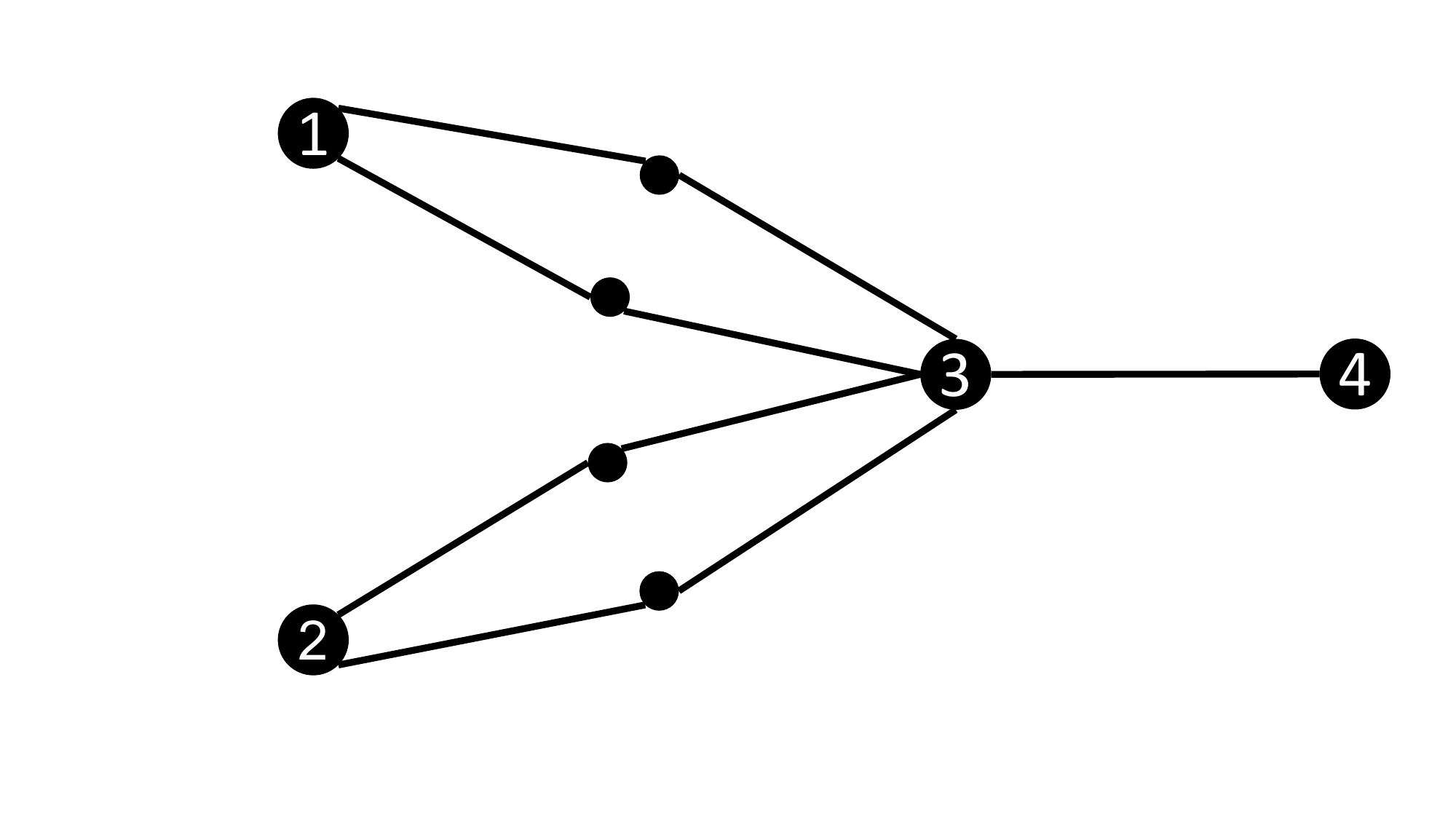}\label{fig:multipl:b}}\\
\subfloat[conversion to a directed graph]{\includegraphics[width=0.8\linewidth,trim={0 2.8cm 0 2.2cm},clip]{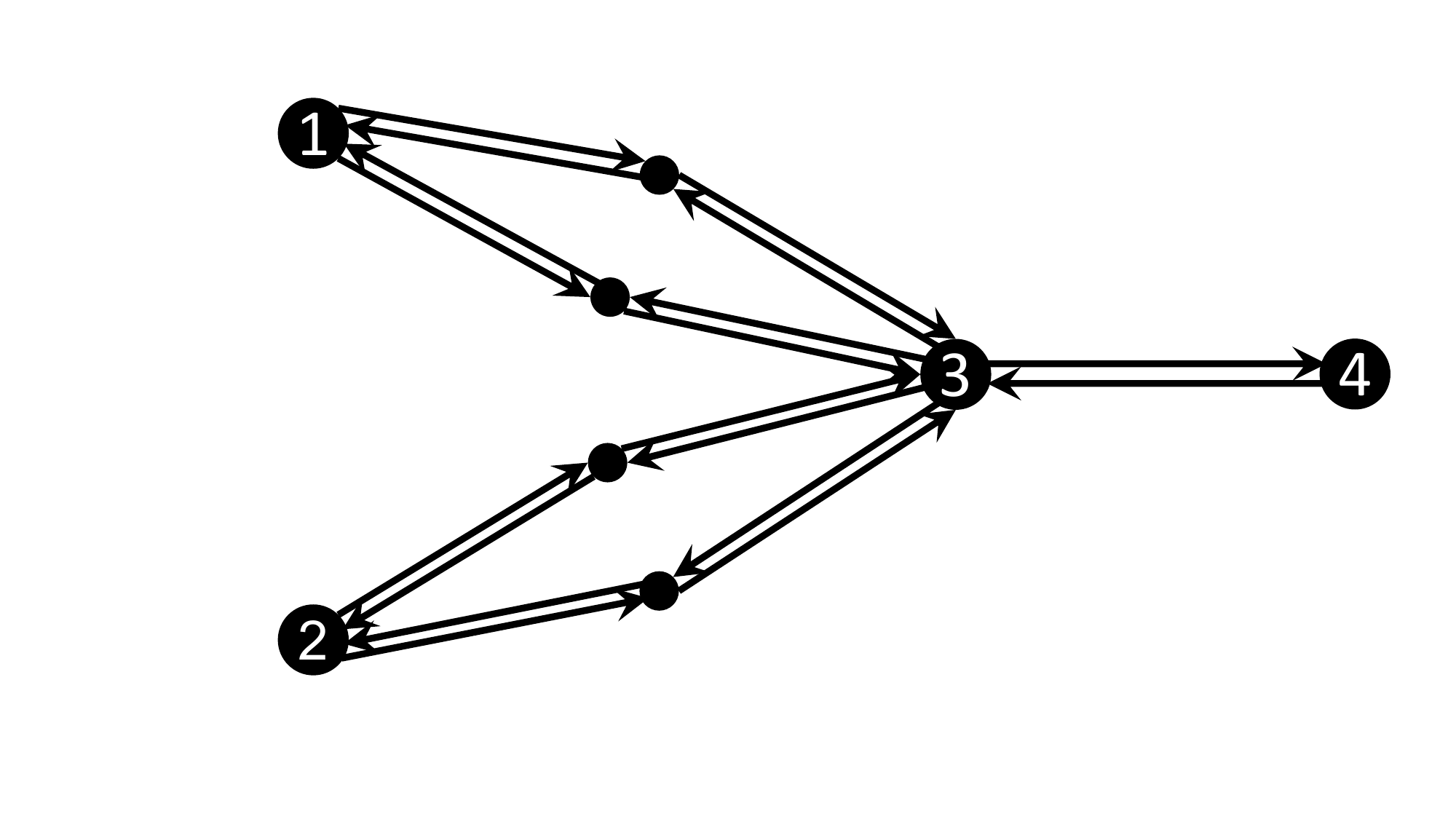}\label{fig:multipl:c}}
\caption{An example multiplexed network snapshot snippet, (a), is transformed by first placing a node of unit gain factor in the middle of each entangled link as shown in (b), and then by converting to a directed graph, (c).}
\label{fig:multipl}
\end{figure}
In this section we extend our approach to compute the capacity of a multiplexed quantum network, \ie, one in which neighboring nodes may establish multiple entangled links in a single time slot. There are two approaches of extending the problem formulation to include multiplexing at the link level. The first is to modify the MIQCP from the previous section to account for the possible presence of more than one Bell pair on each link. Another approach is to instead transform the graph in a way that would conform with the MIQCP in Section \ref{sec:miqcp}. While the first method may be more efficient in yielding a solution (it uses fewer decision variables), in this work, we adapt the latter approach as it has the advantage of not requiring any changes to the MIQCP, thus making it straightforward to check its correctness.

We propose to carry out such a transformation as follows: consider a link that has capacity $c=2,3,\dots$ We first replace the link with $c$ undirected links, each with capacity one. 
We then transform each such entangled link $(i,j)$ into two edges and a new node $n$ with unit gain factor, so that the result is a snapshot with edges $(i,n)$ and $(n,j)$, and no direct edge between $i$ and $j$. An example of this step is shown in Figure~\ref{fig:multipl:b}, for the multiplexed network snippet in \ref{fig:multipl:a}, where all nodes are assumed to be internal ones. This step is only necessary for snapshot edges with capacities of at least two -- notice that the edge $(3,4)$ in the example is unchanged. Note further that this transformation does not change the solution of the problem, in particular because each new node has a unit gain factor and is only connected to two nodes which no longer have a direct connection. Hence, the resulting snapshot is equivalent to the original one within the context of our flow problem, but has the advantageous property that each edge is of unit capacity. The next step is thus to convert the transformed snapshot into a directed graph, as shown in Figure~\ref{fig:multipl:c} (for a snippet of a snapshot), enabling us to apply the MIQCP to obtain a solution.

\begin{figure}[t]
\centering
\includegraphics[width=0.8\linewidth,trim={0 3cm 0 1.9cm},clip]{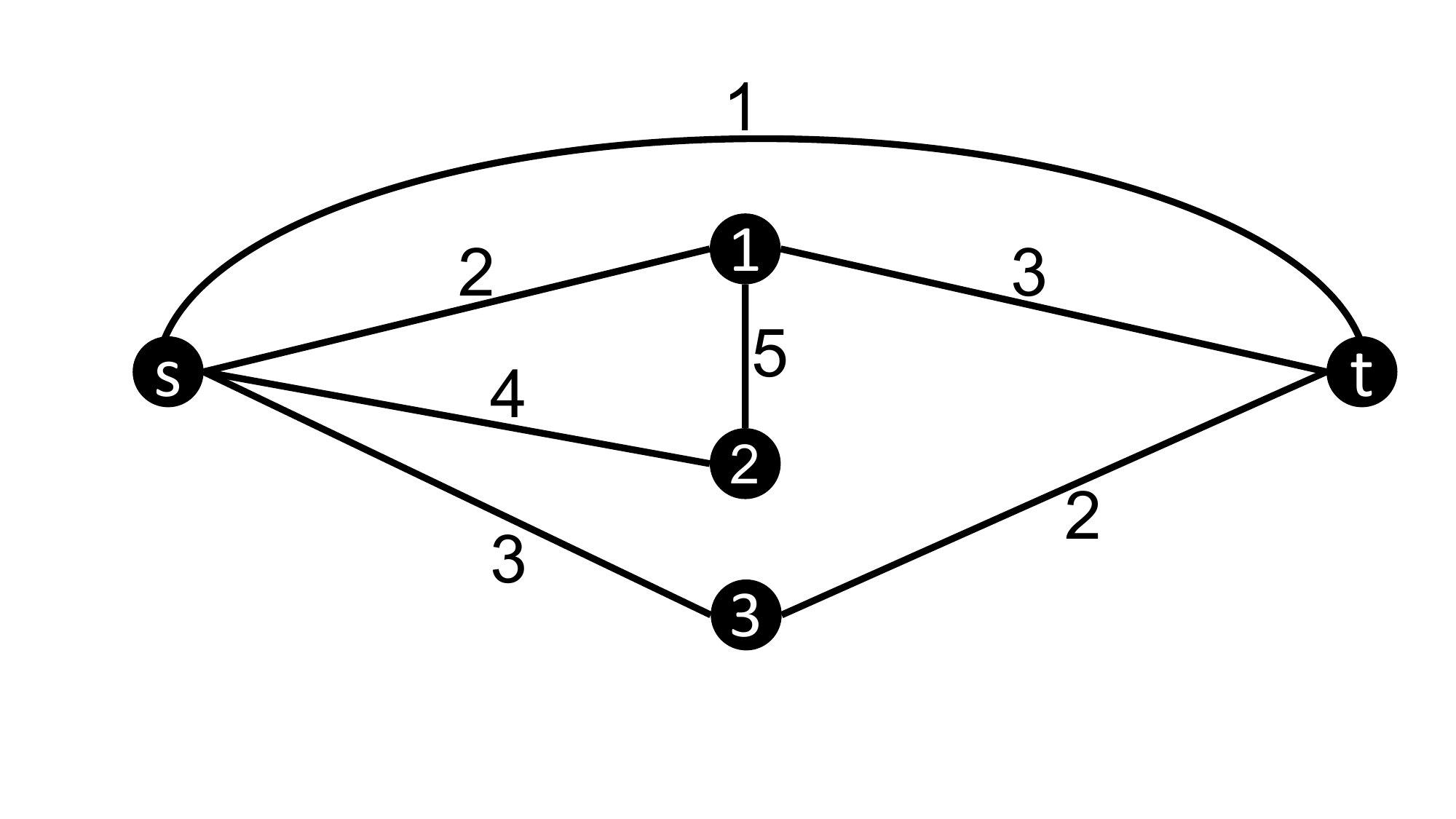}
\caption{A five-node network snapshot with multiplexing. Each edge has a capacity, or the number of available Bell pairs that were successfully generated for this snapshot. The loss factors for the internal nodes $1$, $2$, and $3$ are $0.5$, $0.27$, and $0.64$, respectively.}
\label{fig:simplenetMultiplex}
\end{figure}
Figure \ref{fig:simplenetMultiplex} shows an example of a simple network snapshot with multiplexing.
 The solution for this example is easy to obtain by hand: first, note that there is a direct edge between $s$ and $t$, of capacity 1. Then, there are two paths $s-3-t$, contributing $0.64\times 2 = 1.28$ to the total capacity. The next obvious paths are the two $s-1-t$ paths (each of which necessarily yields a higher capacity than the longer path $s-2-1-t$), with an overall capacity of $0.5\times 2 = 1$. Finally, we may use one of the $s-2-1-t$ paths, with a capacity of $0.27\times 0.5=0.135$. This yields a snapshot capacity of $1+1.28+1+0.135=3.415$, which is also the output of the MIQCP.

\section{Numerical Examples}
\label{sec:numexamples}
We evaluate our capacity computation approach on several unit edge capacity and multiplexed quantum network scenarios. For non-multiplexed networks, we validate the results of the MIQCP using a brute-force capacity computation method applied to every possible network snapshot. For multiplexed networks, whose brute-force validation is infeasible due to both the total number of snapshots as well as their complexities (\ie, a large number of disjoint path set combinations to be evaluated), we apply the brute-force method to a small subset of the snapshots. These results are further validated by the capacity predictions of the local-knowledge based algorithm introduced by Pant \textit{et. al}, \cite{pant2019routing}. In this work, the authors study multi-path entanglement routing in homogeneous networks, wherein all links and nodes have identical entanglement generation and swapping success probabilities, respectively. The proposed routing algorithm uses a Euclidean distance based metric, as well as local link state knowledge to inform each individual repeater what entanglement swaps to perform. To compare against our formulation, the distance metric was updated to use hop distances which was seen to perform better on our arbitrary-topology, non-homogeneous networks. Additional modifications were made to account for multiplexed edges, before Monte Carlo simulations were performed to find the average capacity; a full description of the algorithm is given in \cite{emilyTMPaper}. As this algorithm provides no guarantee of yielding an optimal solution on a network snapshot, we expect a worse performance than that achieved with the MIQCP.

For most of our examples, entanglement generation success probabilities are computed using the formula $p = c \eta$, where $\eta$ is the transmissivity of optical fiber and $c$ is a factor that accounts for various losses other than the transmission loss in fiber. For a specific link $l$, the former is given by
\begin{align}
\eta_l = 10^{-0.1\beta L_l},
\end{align}
where $\beta$ is the fiber attenuation coefficient and $L_l$ is the length of the link. We take $c=0.9$ and $\beta=0.2$ dB/km. Note that this value for $c$ is  optimistic for near-term quantum networks, but using it does allow us to explore more interesting regimes (as well as to avoid numerical instability issues).

Two of the networks which we study in this section -- the Abilene and the NSFNet, which formerly stretched across the contiguous U.S., -- are scaled down when we evaluate our capacity computation method on them: the NSFNet to a metropolitan area, and the Abilene to a campus-sized area. The reason for this is threefold: 
\begin{itemize}[noitemsep,nolistsep]
\item[(1)] a first- or second-generation quantum network with this topology stretched across the U.S. would require additional repeater chains to implement long connections, as direct connections over these distances would result in prohibitively high transmission losses;
\item[(2)] ignoring the practical aspect of the item above, high transmission losses translate to very small entanglement generation success probabilities, potentially resulting in numerical instability during capacity computation; and 
\item[(3)] studying metropolitan and campus area sized networks is more relevant to the near-term.
\end{itemize}
We solve each snapshot's corresponding MIQCP using the Gurobi Optimizer (version 9.1.2, with default solver options) \cite{gurobi}. Our code is available at \cite{github_repo}.
Throughout this section, network capacities are in ebits per second, unless otherwise noted.

\subsection{Simple Five-Node Network}
The first scenario we consider is the simple five-node network in Figure~\ref{fig:simplenetMultiplex}. The link-level entanglement generation success probabilities for this network have been randomly chosen, and are $0.5679$, $0.5179$, $0.4723$, $0.2479$, $0.7839$, $0.5423$, $0.4308$ for edges $(0,1)$, $(0,2)$, $(0,3)$, $(0,4)$, $(1,2)$, $(1,4)$, and $(3,4)$, respectively. The entanglement swapping success probabilities, also randomly chosen, are $0.5$, $0.27$, $0.64$ for nodes $1$, $2$ and $3$, respectively. The network capacity of this network is $1.2121$ ebits per time unit, which can be compared to the maximum possible $3.415$ when all links are present. The network capacity value obtained via MIQCP matches that of the brute-force algorithm. The local knowledge algorithm produces an average capacity of $1.18019$, similar to our calculated results. 

\subsection{Networks Based on the Abilene Network}
\begin{figure}
\centering
\includegraphics[width=0.9\linewidth]{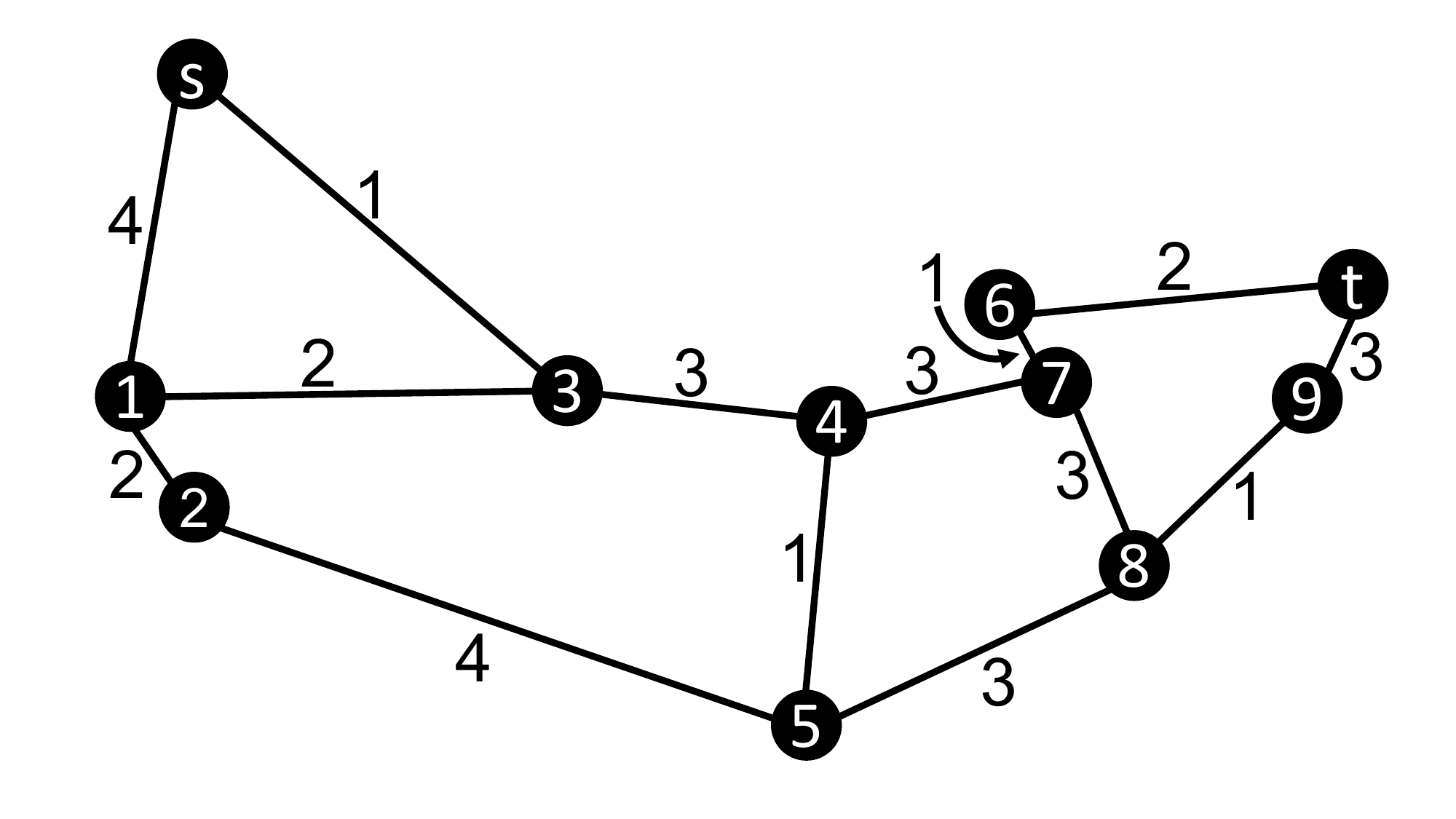}
\caption{The Abilene network topology, scaled down to a metropolitan area (see Table \ref{tab:AbileneDescr} for link lengths and entanglement generation success probabilities). Edges are labeled with link capacities for the multiplexed version of the network.}
\label{fig:netAbilene}
\end{figure}
\begin{table}
\centering
\begin{tabular}{l|c|c|c}
Edge & Original dist. (km) & Scaled dist. (km) & Success prob.\\
\hline
$(s,1)$ & 1138 & 1.138 & 0.8540\\
$(s,3)$ & 1641 & 1.641 & 0.8345\\
$(1,2)$ & 503 & 0.503 & 0.8794\\
$(1,3)$ & 1504 & 1.504 & 0.8398\\
$(2,5)$ & 2206 & 2.206 & 0.8131\\
$(3,4)$ & 896 & 0.896 & 0.8636\\
$(4,5)$ &1041 & 1.041 & 0.8579\\
$(4,7)$ & 727 & 0.727 & 0.8704\\
$(5,8)$ & 1128 & 1.128 & 0.8544\\
$(6,7)$ & 265 & 0.265 & 0.8891\\
$(6,t)$ & 1144 & 1.144 & 0.8538\\
$(7,8)$ & 688 & 0.688 & 0.8719\\
$(8,9)$ & 872 & 0.872 & 0.8646\\
$(9,t)$ &328 & 0.328 & 0.8865
\end{tabular}
\caption{Description of the Abilene Network, with success probabilities corresponding to the scaled-down distances.}
\label{tab:AbileneDescr}
\end{table}
Launched in the late nineties, the Abilene Network connected several institutions across the United States and was used extensively for research and development until its retirement in 2007. The topology of the network is shown in Figure~\ref{fig:netAbilene}; Table~\ref{tab:AbileneDescr} provides a more detailed description of the network.
Each node's BSM success probability was sampled uniformly at random from $(0.9,1)$, resulting in $\{q_1,\dots,q_9\}=\{0.99,~0.96,~0.92,~0.93,~0.91,~0.99,~0.97,~0.92,~0.98\}$.

The capacity of the non-multiplexed scaled-down Abilene Network, as defined by Eq. (\ref{eq:capacity}), is $0.8301$ (this number matches that of the brute-force search algorithm output). This value can be contrasted with the capacity of the snapshot with all entangled links present -- $1.607$. The local knowledge algorithm gives an average capacity of $ 0.85271$. For the multiplexed Abilene Network, we leave the link generation probabilities unchanged, \ie, in each time slot, there may be up to $c_l$ EPR pairs on link $l$, each generated with probability $p_l$. The swapping success probabilities are also left unchanged. The capacity of the multiplexed Abilene Network with edge capacities $c_l$ as shown in Figure~\ref{fig:netAbilene} is $1.387$. This can be compared to a capacity of $0.95198$ provided by the local-knowledge algorithm. For a multiplexed Abilene Network with $c_l=2$, $\forall l$, capacity increases to $1.983$. This increase in capacity compared to that of the network in Figure~\ref{fig:netAbilene} is due to the possibility of generating extra links on bottleneck edges $(6,7)$ and $(8,9)$. This can be compared to a capacity of $1.62778$ given by the local-knowledge algorithm.
\subsection{Network Based on the NSFNet}
NSFNET was a backbone network that linked several U.S. national supercomputing centers and later played a part in developing a portion of the Internet backbone.
A version of the NSFNET topology is presented in Figure \ref{fig:nsfnet}. Table \ref{tab:NSFNetDescr} describes the original approximate link distances, adopted from \cite{ghose2005multihop}, as well as scaled distances. Success probabilities for link-level entanglement generation correspond to the scaled distances. 
The BSM success probabilities, sampled uniformly at random from $(0.5,1)$ and rounded down, are set to $\{q_1,\dots,q_{12}\} = \{0.9, 0.9, 0.8, 0.5, 0.5, 0.5, 0.5, 0.5, 0.7, 0.5, 0.7, 0.5\}$.
The network capacity is $0.1013397$ (brute-force algorithm validates this). The local knowledge algorithm gives an average capacity of $0.08136$.
\begin{figure}[t!]
\centering
\includegraphics[width=0.9\linewidth]{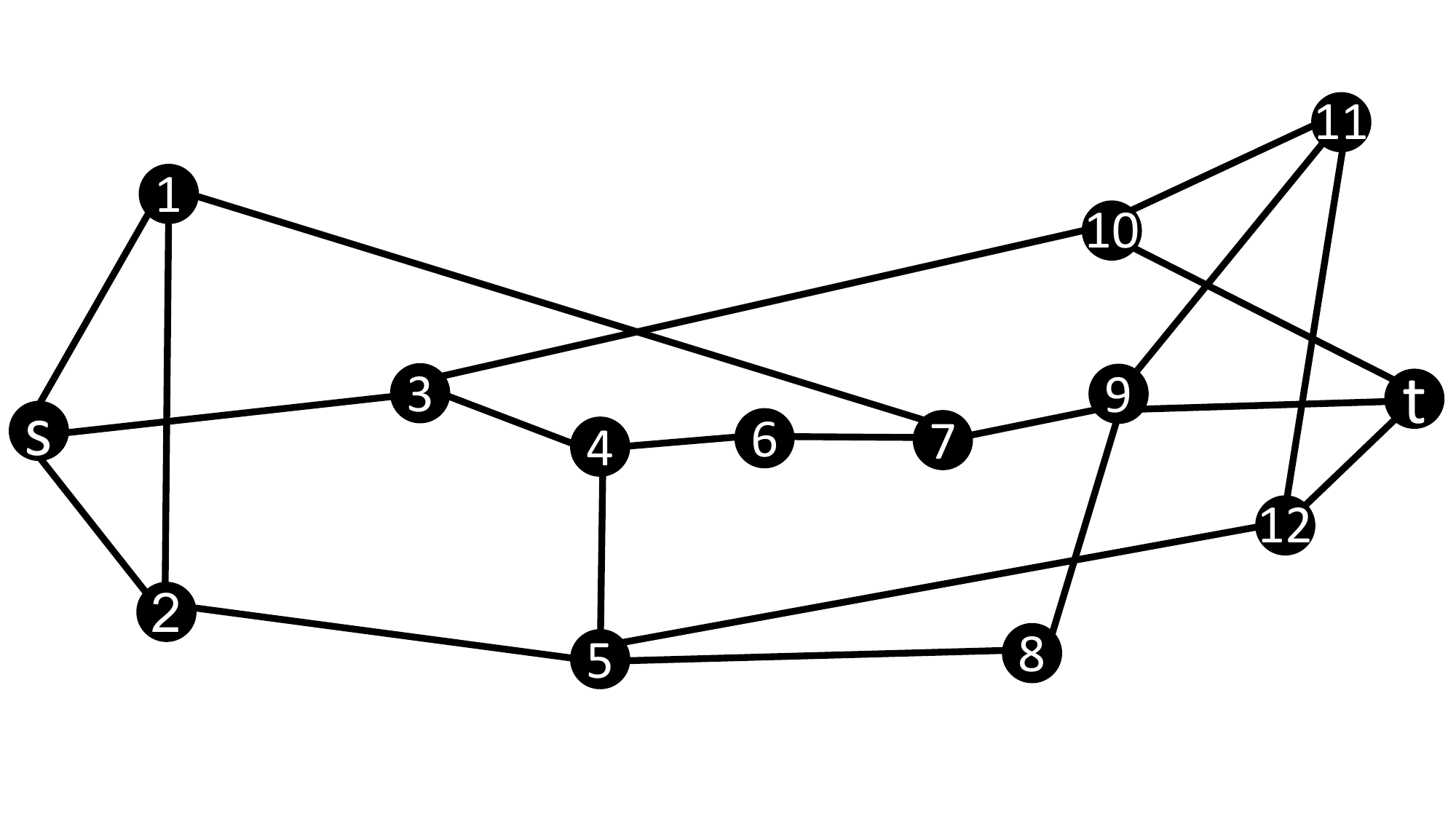}
\caption{The NSFNet topology.}
\label{fig:nsfnet}
\end{figure}
\begin{table}[t]
\centering
\begin{tabular}{l|c|c|c}
Edge & Original dist. (km) & Scaled dist. (km) & Success prob.\\
\hline
$(s,1)$     & 1100  & 11 & 0.5423\\
$(s,2)$     &  600 & 6 & 0.6827\\
$(s,3)$     & 1000 & 10 & 0.5679\\
$(1,2)$ & 1600 & 16 & 0.4308\\
$(1,7)$ & 2800 & 28 & 0.2479\\
$(2,5)$ & 2000 & 20 & 0.3583\\
$(3,4)$ & 600 & 6 & 0.6827\\
$(3,{10})$ & 2400 & 24 & 0.2980\\
$(4,5)$ & 1100 & 11 & 0.5423\\
$(4,6)$ & 800 & 8 & 0.6226\\
$(5,8)$ & 1200 & 12 & 0.5179\\
$(5,{12})$ & 2000 & 20 & 0.3583\\
$(6,7)$ & 700 & 7 & 0.6520\\
$(7,9)$ & 700 & 7 & 0.6520\\
$(8,9)$ & 900 & 9 & 0.5946\\
$(9,{11})$ & 500 & 5 & 0.7149\\
$(9,t)$ & 500 & 5 & 0.7149\\
$({10},{11})$ & 800 & 8 & 0.6226\\
$({10},t)$ & 1000 & 10 & 0.5679\\
$({11},{12})$ & 500 & 5 & 0.7149\\
$({12},t)$ & 300 & 3 & 0.7839
\end{tabular}
\caption{Description of the NSFNet network, with success probabilities corresponding to rescaled metropolitan-area distances.}
\label{tab:NSFNetDescr}
\end{table}
\subsection{Network Based on SURFNet}
\begin{figure*}[t]
\centering
\includegraphics[width=0.8\linewidth]{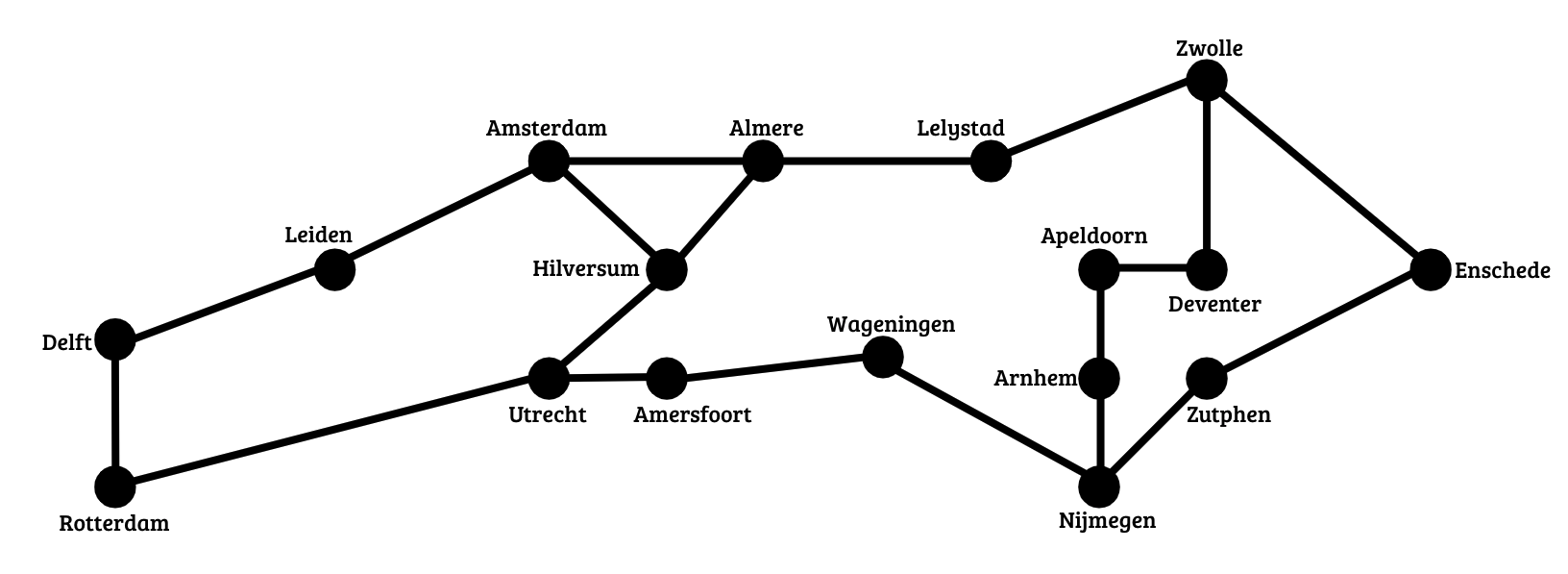}
\caption{A pruned version of the SURFnet topology.}
\label{fig:surfnet}
\end{figure*}
\begin{table}[t]
\centering
\begin{tabular}{c|c|c}
Edge & Distance (km) &  Success prob.\\
\hline
(Delft, Rotterdam) $\equiv (s,1)$    & 16.8 & 0.4152\\
(Delft, Leiden) $\equiv (s,2)$& 30.6 & 0.2199\\
(Leiden, Amsterdam) $\equiv (2,3)$& 60.4 & 0.0557\\
(Rotterdam, Utrecht) $\equiv (1,5)$& 70 & 0.0358\\
(Amsterdam, Hilversum) $\equiv (3,4)$ & 30.2 & 0.2240\\
(Amsterdam, Almere) $\equiv (3,7)$ & 38.9 & 0.1501\\
(Utrecht, Hilversum) $\equiv (4,5)$ & 36.7 & 0.1661\\
(Hilversum,  Almere) $\equiv (4,7)$ & 35.4 & 0.1763\\
(Utrecht, Amersfoort) $\equiv (5,6)$ & 33.8 & 0.1898\\
(Almere, Lelystad) $\equiv (7,9)$& 44.2 & 0.1176\\
(Amersfoort, Wageningen) $\equiv (6,8)$& 62.5 & 0.0506\\
(Wageningen, Nijmegen) $\equiv (8,12)$& 66.3 & 0.0425\\
(Nijmegen, Arnhem) $\equiv (11,12)$ & 25.7 & 0.2756\\
(Nijmegen, Zutphen) $\equiv (12,15)$ & 58.1 & 0.0620\\
(Arnhem, Apeldoorn) $\equiv (10,11)$& 45.3 & 0.1117\\
(Apeldoorn, Deventer) $\equiv (10,14)$& 24.4 & 0.2926\\
(Lelystad, Zwolle) $\equiv (9,13)$& 47.7 & 0.1001\\
(Deventer, Zwolle) $\equiv (13,14)$& 44.7 & 0.1149\\
(Zwolle, Enschede) $\equiv (13,t)$ & 78.7 & 0.0240\\
(Zutphen, Enschede) $\equiv (15,t)$& 60 & 0.0568
\end{tabular}
\caption{Link lengths and corresponding link-level entanglement generation success probabilities for the SURF network.}
\label{tab:SURFnetDescr}
\end{table}
SURFnet is a backbone network used for research and education purposes in the Netherlands.
Figure \ref{fig:surfnet} depicts a pruned version of the SURFnet topology. The full topology is provided in \cite{rabbie2022designing} and a detailed dataset is located at \cite{repalloc}. Note that, contrary to what is shown in  Figure 2 of \cite{rabbie2022designing}, according to SURFnet topological data there is no link between nodes \emph{Delft 1} and \emph{R'dam 1}: to make the (\emph{Delft}, \emph{Rotterdam}) link, we assume there is no \emph{Delft 2} node, and sum the fiber lengths between \emph{Delft 1}, \emph{Delft 2}, and \emph{R'dam 1}. In a similar manner, we also eliminate nodes \emph{Utrecht 2}, \emph{Wageningen 1}, \emph{Nijmegen 2}, \emph{Zwolle 2}, and \emph{Enschede 1}. The pruned topology is described in more detail in Table~\ref{tab:SURFnetDescr}. Swapping success probabilities  are set to $\{q_1,\dots,q_{15}\}=
\{0.87, 0.74, 0.79, 0.62, 0.73, 0.98, 0.77, 0.76, 0.62, 0.74, 0.81,\\ 0.84, 0.7, 0.68, 0.99\}$.
Our goal is to determine the $\st$ flow between Delft and Enschede. We determine this value to be $1.0762\times10^{-7}$ (validated by the brute-force algorithm). The local knowledge based algorithm gave a network capacity of 0. This could be in part because it was calculated over a finite number of samples, and is believed to agree with our results. We remark that the reason for the low capacity compared to the other examples is due to both the comparably large number of hops between $s$ and $t$, as well as the (in general) much lower link-level entanglement generation rates. Based on these results, we conclude that this network would benefit from multiplexing, or if unavailable, an alternative method of boosting link-level rates, \eg, via a single-photon entanglement generation scheme \cite{cabrillo1999creation,humphreys2018deterministic}.

\section{Conclusion}
\label{sec:concl}
In this work we studied the problem of multipath entanglement routing between two nodes in a quantum network. The links in the network may have multiplexing capabilities, but repeater nodes may not be able to perform deterministic entanglement swapping. To solve the problem, we proposed an MIQCP, the solution of which yields the optimal throughput for a given time interval. By averaging over all possible network states, we were able to obtain the average capacity of various networks included in our case studies.

A possible future direction of our work is to explore approximation methods to the capacity computation. A simple way to achieve this would be to compute the capacities of only the most likely network snapshots (those occurring with relatively high probabilities). Other approximations, possibly with much lower time complexities, are also of interest. An extension of our work would be to consider a time horizon beyond one time slot, and to allow the repeater nodes to perform $n-$GHZ measurements, for $n>2$ instead of, or in addition to BSMs. 

\appendices
\if{false}
\section{Comparison of Snapshot Capacity Formulation to a Rate Maximization Problem}
\begin{figure}
    \centering
    \includegraphics[width=0.8\linewidth]{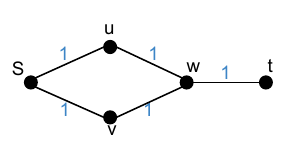}
    \caption{Example quantum network snapshot. All links have unit capacity, and internal nodes $u$, $v$, $w$ succeed at swapping with probability 0.5.}
    \label{fig:comp_example}
\end{figure}
\begin{figure}
    \centering
    \includegraphics[width=0.98\linewidth]{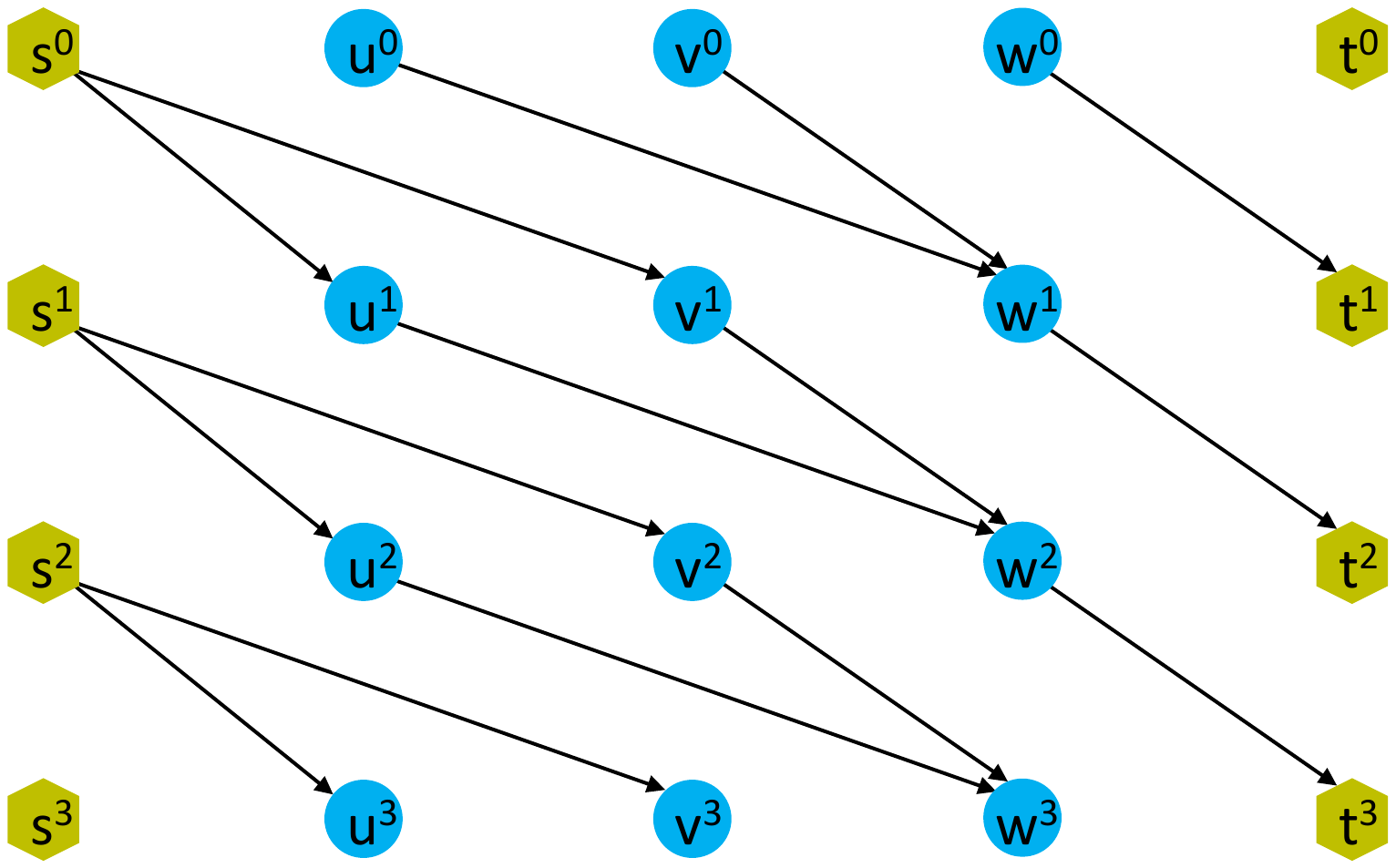}
    \caption{Extended network graph corresponding to the original snapshot depicted in Figure \ref{fig:comp_example}. This graph is used as an input to an edge-based LP formulation that maximizes the entanglement generation rate between nodes $s$ and $t$.}
    \label{fig:comp_example_edgebased}
\end{figure}
In this appendix, we compare the formulation introduced in this text to that of \cite{chakraborty2020entanglement}. There, the authors constructed a Linear Program (LP) to solve the entanglement distribution problem in a network with multiple source-destination pairs. It was assumed that entanglement swapping succeeds probabilistically (albeit, all repeaters succeed with equal probability), and that users have minimum end-to-end fidelity requirements which translates into a maximum number of swapping operations. To compare this work with our setup, we will therefore consider the network depicted in Figure~\ref{fig:comp_example} with a single source-destination pair $(s,t)$ and $q_i=0.5$, $\forall~i\in V\setminus \{s,t\}$. For simplicity, we set all link capacities to one and assume that the maximum allowable path length, $l_{\max}$, is three.

The rate maximization problem is solved in \cite{chakraborty2020entanglement} using an edge-based formulation which involves the construction of an extended graph. This graph consists of $l_{\max}+1$ copies of each node $n$, labelled $n^0,\dots,n^{l_{\max}}$, and unit-capacity edges $(n^i,m^{i+1})$, $0\leq i < l_{\max}$ whenever there exists an edge between nodes $n$ and $m$ in the original graph. Figure \ref{fig:comp_example_edgebased} presents the extended graph of the one shown in Figure \ref{fig:comp_example}.
This graph is then used as an input to the LP that maximizes the flow from source $s^0$ to destinations $t^i$, $1\leq i \leq l_{\max}$. For our example, the optimal value that results from this formulation is 0.5: using Figure \ref{fig:comp_example_edgebased}, we identify two paths from $s^0$ to $t^3$:
\begin{align}
    s^0 &\to u^1 \to w^2 \to t^3,\\
    s^0 &\to v^1 \to w^2 \to t^3.
\end{align}
Edges $(s^0,u^1)$ and $(s^0,v^1)$ carry a unit of flow each, 

\gv{Work in progress.}
\fi
\section*{Acknowledgment}
DT and SG acknowledge support from the NSF grant CNS-1955744, NSF-ERC Center for Quantum Networks grant EEC-1941583, and MURI ARO Grant W911NF2110325. This work is supported in part by QuTech NWO funding 2020–2024—Part I ‘Fundamental Research’, Project Number 601.QT.001-1, financed by the Dutch Research Council (NWO). GV acknowledges support from the NWO ZK QSC Ada Lovelace Fellowship and NWO QSC grant  BGR2 17.269. 
GV thanks Subhransu Maji, Guus Avis, and Ashlesha Patil for useful discussions. This work was performed in part using high performance computing equipment obtained under a grant from the Collaborative R\&D Fund managed by the Massachusetts Technology Collaborative.
\bibliographystyle{IEEEtran}
\bibliography{refs}

\begin{thebibliography}{10}
\providecommand{\url}[1]{#1}
\csname url@samestyle\endcsname
\providecommand{\newblock}{\relax}
\providecommand{\bibinfo}[2]{#2}
\providecommand{\BIBentrySTDinterwordspacing}{\spaceskip=0pt\relax}
\providecommand{\BIBentryALTinterwordstretchfactor}{4}
\providecommand{\BIBentryALTinterwordspacing}{\spaceskip=\fontdimen2\font plus
\BIBentryALTinterwordstretchfactor\fontdimen3\font minus
  \fontdimen4\font\relax}
\providecommand{\BIBforeignlanguage}[2]{{%
\expandafter\ifx\csname l@#1\endcsname\relax
\typeout{** WARNING: IEEEtran.bst: No hyphenation pattern has been}%
\typeout{** loaded for the language `#1'. Using the pattern for}%
\typeout{** the default language instead.}%
\else
\language=\csname l@#1\endcsname
\fi
#2}}
\providecommand{\BIBdecl}{\relax}
\BIBdecl

\bibitem{munro2015inside}
W.~J. Munro, K.~Azuma, K.~Tamaki, and K.~Nemoto, ``Inside quantum repeaters,''
  \emph{IEEE Journal of Selected Topics in Quantum Electronics}, vol.~21,
  no.~3, pp. 78--90, 2015.

\bibitem{muralidharan2016optimal}
S.~Muralidharan, L.~Li, J.~Kim, N.~L{\"u}tkenhaus, M.~D. Lukin, and L.~Jiang,
  ``Optimal architectures for long distance quantum communication,''
  \emph{Scientific reports}, vol.~6, no.~1, pp. 1--10, 2016.

\bibitem{bennet1984quantum}
C.~H. Bennett and G.~Brassard, ``Quantum cryptography: Public key distribution
  and coin tossing,'' in \emph{Proceedings of the IEEE International Conference
  on Computers, Systems, and Signal Processing, Bangalore, Dec. 1984}, 1984,
  pp. 175--179.

\bibitem{ekert1991quantum}
A.~K. Ekert, ``{Quantum cryptography based on Bell’s theorem},''
  \emph{Physical review letters}, vol.~67, no.~6, p. 661, 1991.

\bibitem{bennett1992quantum}
C.~H. Bennett, G.~Brassard, and N.~D. Mermin, ``{Quantum cryptography without
  Bell’s theorem},'' \emph{Physical review letters}, vol.~68, no.~5, p. 557,
  1992.

\bibitem{broadbent2009universal}
A.~Broadbent, J.~Fitzsimons, and E.~Kashefi, ``Universal blind quantum
  computation,'' in \emph{2009 50th Annual IEEE Symposium on Foundations of
  Computer Science}.\hskip 1em plus 0.5em minus 0.4em\relax IEEE, 2009, pp.
  517--526.

\bibitem{leichtle2021verifying}
D.~Leichtle, L.~Music, E.~Kashefi, and H.~Ollivier, ``Verifying bqp
  computations on noisy devices with minimal overhead,'' \emph{PRX Quantum},
  vol.~2, no.~4, p. 040302, 2021.

\bibitem{duan2001long}
L.-M. Duan, M.~D. Lukin, J.~I. Cirac, and P.~Zoller, ``Long-distance quantum
  communication with atomic ensembles and linear optics,'' \emph{Nature}, vol.
  414, no. 6862, pp. 413--418, 2001.

\bibitem{maier2020investigating}
D.~Maier, ``Investigating the scalability of quantum repeater protocols based
  on atomic ensembles,'' Master's thesis, 2020.

\bibitem{rabbie2020simulation}
J.~Rabbie, ``Simulation model for atomic ensemble based quantum repeaters and
  the optimization of their positioning,'' Master's thesis, 2020.

\bibitem{caleffi2017optimal}
M.~Caleffi, ``Optimal routing for quantum networks,'' \emph{IEEE Access},
  vol.~5, pp. 22\,299--22\,312, 2017.

\bibitem{van2013path}
R.~Van~Meter, T.~Satoh, T.~D. Ladd, W.~J. Munro, and K.~Nemoto, ``Path
  selection for quantum repeater networks,'' \emph{Networking Science}, vol.~3,
  pp. 82--95, 2013.

\bibitem{gyongyosi2017entanglement}
L.~Gyongyosi and S.~Imre, ``Entanglement-gradient routing for quantum
  networks,'' \emph{Scientific reports}, vol.~7, no.~1, pp. 1--14, 2017.

\bibitem{schoute2016shortcuts}
E.~Schoute, L.~Mancinska, T.~Islam, I.~Kerenidis, and S.~Wehner, ``Shortcuts to
  quantum network routing,'' \emph{arXiv preprint arXiv:1610.05238}, 2016.

\bibitem{gyongyosi2018decentralized}
L.~Gyongyosi and S.~Imre, ``Decentralized base-graph routing for the quantum
  internet,'' \emph{Physical Review A}, vol.~98, no.~2, p. 022310, 2018.

\bibitem{chakraborty2019distributed}
K.~Chakraborty, F.~Rozpedek, A.~Dahlberg, and S.~Wehner, ``Distributed routing
  in a quantum internet,'' \emph{arXiv preprint arXiv:1907.11630}, 2019.

\bibitem{pirandola2019end}
S.~Pirandola, ``End-to-end capacities of a quantum communication network,''
  \emph{Communications Physics}, vol.~2, no.~1, p.~51, 2019.

\bibitem{azuma2016fundamental}
K.~Azuma, A.~Mizutani, and H.-K. Lo, ``{Fundamental rate-loss trade-off for the
  quantum internet},'' \emph{Nature communications}, vol.~7, no.~1, p. 13523,
  2016.

\bibitem{harney2022end}
C.~Harney and S.~Pirandola, ``{End-to-end capacities of imperfect-repeater
  quantum networks},'' \emph{Quantum Science and Technology}, vol.~7, no.~4, p.
  045009, 2022.

\bibitem{azuma2017aggregating}
K.~Azuma and G.~Kato, ``{Aggregating quantum repeaters for the quantum
  internet},'' \emph{Physical Review A}, vol.~96, no.~3, p. 032332, 2017.

\bibitem{bauml2020linear}
S.~B{\"a}uml, K.~Azuma, G.~Kato, and D.~Elkouss, ``{Linear programs for
  entanglement and key distribution in the quantum internet},''
  \emph{Communications Physics}, vol.~3, no.~1, p.~55, 2020.

\bibitem{dai2020optimal}
W.~Dai, T.~Peng, and M.~Z. Win, ``Optimal remote entanglement distribution,''
  \emph{IEEE Journal on Selected Areas in Communications}, vol.~38, no.~3, pp.
  540--556, 2020.

\bibitem{cicconetti2021request}
C.~Cicconetti, M.~Conti, and A.~Passarella, ``Request scheduling in quantum
  networks,'' \emph{IEEE Transactions on Quantum Engineering}, vol.~2, pp.
  2--17, 2021.

\bibitem{le2022dqra}
L.~Le and T.~N. Nguyen, ``{DQRA: Deep Quantum Routing Agent for Entanglement
  Routing in Quantum Networks},'' \emph{IEEE Transactions on Quantum
  Engineering}, vol.~3, pp. 1--12, 2022.

\bibitem{victora2020purification}
M.~Victora, S.~Tserkis, S.~Krastanov, A.~S. de~la Cerda, S.~Willis, and
  P.~Narang, ``Entanglement purification on quantum networks,'' \emph{Physical
  Review Research}, vol.~5, no.~3, Sep. 2023.

\bibitem{yang2022online}
L.~Yang, Y.~Zhao, H.~Xu, and C.~Qiao, ``Online entanglement routing in quantum
  networks,'' in \emph{2022 IEEE/ACM 30th International Symposium on Quality of
  Service (IWQoS)}.\hskip 1em plus 0.5em minus 0.4em\relax IEEE, 2022, pp.
  1--10.

\bibitem{nguyen2022multiple}
T.~N. Nguyen, K.~J. Ambarani, L.~Le, I.~Djordjevic, and Z.-L. Zhang, ``A
  multiple-entanglement routing framework for quantum networks,'' \emph{arXiv
  preprint arXiv:2207.11817}, 2022.

\bibitem{zeng2022multi}
Y.~Zeng, J.~Zhang, J.~Liu, Z.~Liu, and Y.~Yang, ``Multi-entanglement routing
  design over quantum networks,'' in \emph{IEEE INFOCOM 2022 - IEEE Conference
  on Computer Communications}.\hskip 1em plus 0.5em minus 0.4em\relax IEEE, May
  2022.

\bibitem{zhang2022multipath}
L.~Zhang, S.-X. Ye, Q.~Liu, and H.~Chen, ``Multipath concurrent entanglement
  routing in quantum networks based on virtual circuit,'' in \emph{2022 4th
  International Conference on Advances in Computer Technology, Information
  Science and Communications (CTISC)}.\hskip 1em plus 0.5em minus 0.4em\relax
  IEEE, 2022, pp. 1--5.

\bibitem{nguyen2022maximizing}
T.~N. Nguyen, D.~H. Nguyen, D.~H. Pham, B.-H. Liu, and H.~N. Nguyen, ``{LP
  Relaxation-Based Approximation Algorithms for Maximizing Entangled Quantum
  Routing Rate},'' in \emph{ICC 2022-IEEE International Conference on
  Communications}.\hskip 1em plus 0.5em minus 0.4em\relax IEEE, 2022, pp.
  3269--3274.

\bibitem{amer2020efficient}
O.~Amer, W.~O. Krawec, and B.~Wang, ``Efficient routing for quantum key
  distribution networks,'' in \emph{2020 IEEE International Conference on
  Quantum Computing and Engineering (QCE)}.\hskip 1em plus 0.5em minus
  0.4em\relax IEEE, 2020, pp. 137--147.

\bibitem{zhao2021redundant}
Y.~Zhao and C.~Qiao, ``Redundant entanglement provisioning and selection for
  throughput maximization in quantum networks,'' in \emph{IEEE INFOCOM
  2021-IEEE Conference on Computer Communications}.\hskip 1em plus 0.5em minus
  0.4em\relax IEEE, 2021, pp. 1--10.

\bibitem{ghaderibaneh2022pre}
M.~Ghaderibaneh, H.~Gupta, C.~Ramakrishnan, and E.~Luo, ``Pre-distribution of
  entanglements in quantum networks,'' in \emph{2022 IEEE International
  Conference on Quantum Computing and Engineering (QCE)}.\hskip 1em plus 0.5em
  minus 0.4em\relax IEEE, 2022, pp. 426--436.

\bibitem{li2022fidelity}
J.~Li, M.~Wang, K.~Xue, R.~Li, N.~Yu, Q.~Sun, and J.~Lu, ``Fidelity-guaranteed
  entanglement routing in quantum networks,'' \emph{IEEE Transactions on
  Communications}, vol.~70, no.~10, pp. 6748--6763, 2022.

\bibitem{li2021effective}
C.~Li, T.~Li, Y.-X. Liu, and P.~Cappellaro, ``Effective routing design for
  remote entanglement generation on quantum networks,'' \emph{npj Quantum
  Information}, vol.~7, no.~1, pp. 1--12, 2021.

\bibitem{shi2020concurrent}
S.~Shi and C.~Qian, ``Concurrent entanglement routing for quantum networks:
  Model and designs,'' in \emph{Proceedings of the Annual conference of the ACM
  Special Interest Group on Data Communication on the applications,
  technologies, architectures, and protocols for computer communication}, 2020,
  pp. 62--75.

\bibitem{chakraborty2020entanglement}
K.~Chakraborty, D.~Elkouss, B.~Rijsman, and S.~Wehner, ``Entanglement
  distribution in a quantum network: A multicommodity flow-based approach,''
  \emph{IEEE Transactions on Quantum Engineering}, vol.~1, pp. 1--21, 2020.

\bibitem{pant2019routing}
M.~Pant, H.~Krovi, D.~Towsley, L.~Tassiulas, L.~Jiang, P.~Basu, D.~Englund, and
  S.~Guha, ``Routing entanglement in the quantum internet,'' \emph{npj Quantum
  Information}, vol.~5, no.~1, pp. 1--9, 2019.

\bibitem{patil2022entanglement}
A.~Patil, M.~Pant, D.~Englund, D.~Towsley, and S.~Guha, ``Entanglement
  generation in a quantum network at distance-independent rate,'' \emph{npj
  Quantum Information}, vol.~8, no.~1, p.~51, 2022.

\bibitem{leone2021qunet}
H.~Leone, N.~R. Miller, D.~Singh, N.~K. Langford, and P.~P. Rohde, ``Qunet:
  Cost vector analysis \& multi-path entanglement routing in quantum
  networks,'' \emph{arXiv preprint arXiv:2105.00418}, 2021.

\bibitem{acin2007entanglement}
A.~Ac{\'\i}n, J.~I. Cirac, and M.~Lewenstein, ``Entanglement percolation in
  quantum networks,'' \emph{Nature Physics}, vol.~3, no.~4, pp. 256--259, 2007.

\bibitem{skrzypczyk2021architecture}
M.~Skrzypczyk and S.~Wehner, ``{An Architecture for Meeting Quality-of-Service
  Requirements in Multi-User Quantum Networks},'' \emph{arXiv preprint
  arXiv:2111.13124}, 2021.

\bibitem{pompili2021experimental}
M.~Pompili, C.~Delle~Donne, I.~te~Raa, B.~van~der Vecht, M.~Skrzypczyk,
  G.~Ferreira, L.~de~Kluijver, A.~J. Stolk, S.~L. Hermans, P.~Pawe{\l}czak
  \emph{et~al.}, ``Experimental demonstration of entanglement delivery using a
  quantum network stack,'' \emph{npj Quantum Information}, vol.~8, no.~1, p.
  121, 2022.

\bibitem{kozlowski2019towards}
W.~Kozlowski and S.~Wehner, ``Towards large-scale quantum networks,'' in
  \emph{Proceedings of the Sixth Annual ACM International Conference on
  Nanoscale Computing and Communication}, 2019, pp. 1--7.

\bibitem{emilyTMPaper}
E.~A. Van~Milligen, E.~Jacobson, A.~Patil, G.~Vardoyan, D.~Towsley, and
  S.~Guha, ``{Entanglement Routing over Networks with Time Multiplexed
  Repeaters},'' \emph{arXiv preprint arXiv:2308.15028}, 2023.

\bibitem{gurobi}
\BIBentryALTinterwordspacing
{Gurobi Optimization, LLC}, ``{Gurobi Optimizer Reference Manual},'' 2023.
  [Online]. Available: \url{https://www.gurobi.com}
\BIBentrySTDinterwordspacing

\bibitem{github_repo}
\BIBentryALTinterwordspacing
G.~Vardoyan, ``Quantum network capacity,'' 2023. [Online]. Available:
  \url{https://github.com/gvardoyan/QuantumNetworkCapacity}
\BIBentrySTDinterwordspacing

\bibitem{ghose2005multihop}
S.~Ghose, R.~Kumar, N.~Banerjee, and R.~Datta, ``{Multihop virtual topology
  design in WDM optical networks for self-similar traffic},'' \emph{Photonic
  Network Communications}, vol.~10, no.~2, pp. 199--214, 2005.

\bibitem{rabbie2022designing}
J.~Rabbie, K.~Chakraborty, G.~Avis, and S.~Wehner, ``Designing quantum networks
  using preexisting infrastructure,'' \emph{npj Quantum Information}, vol.~8,
  no.~1, pp. 1--12, 2022.

\bibitem{repalloc}
\BIBentryALTinterwordspacing
J.~Rabbie and G.~Avis, ``{RepAlloc},'' 2020. [Online]. Available:
  \url{https://github.com/jtrabbie/RepAlloc}
\BIBentrySTDinterwordspacing

\bibitem{cabrillo1999creation}
C.~Cabrillo, J.~I. Cirac, P.~Garcia-Fernandez, and P.~Zoller, ``Creation of
  entangled states of distant atoms by interference,'' \emph{Physical Review
  A}, vol.~59, no.~2, p. 1025, 1999.

\bibitem{humphreys2018deterministic}
P.~C. Humphreys, N.~Kalb, J.~P. Morits, R.~N. Schouten, R.~F. Vermeulen, D.~J.
  Twitchen, M.~Markham, and R.~Hanson, ``Deterministic delivery of remote
  entanglement on a quantum network,'' \emph{Nature}, vol. 558, no. 7709, pp.
  268--273, 2018.

\end{thebibliography}
\EOD
\end{document}